\documentclass{siamart1116}

\usepackage{lipsum}
\usepackage{amsfonts}
\usepackage{graphicx}
\usepackage{epstopdf}
\usepackage{algorithmic}
\ifpdf
  \DeclareGraphicsExtensions{.eps,.pdf,.png,.jpg}
\else
  \DeclareGraphicsExtensions{.eps}
\fi

\numberwithin{theorem}{section}

\newcommand{\TheTitle}{Plasmonic interaction between nanospheres} 
\newcommand{\TheAuthors}{Sanghyeon Yu and Habib Ammari}

\headers{\TheTitle}{\TheAuthors}

\title{{\TheTitle}
}

\author{ Sanghyeon Yu\thanks{Department of Mathematics, ETH Z\"urich, R\"amistrasse 101, CH-8092 Z\"urich, Switzerland (\email{sanghyeon.yu@sam.math.ethz.ch},\email{habib.ammari@math.ethz.ch}).} \and Habib Ammari\footnotemark[1] }

\usepackage{amsopn}

\def\ep{\epsilon}

\newcommand{\Bc}{\mathbf{c}}
\newcommand{\ds}{\displaystyle}
\newcommand{\beq}{\begin{equation}}
\newcommand{\eeq}{\end{equation}}
\newcommand{\Br}{\mathbf{r}}

\ifpdf
\hypersetup{
  pdftitle={\TheTitle},
  pdfauthor={\TheAuthors}
}
\fi

\begin{document}

\maketitle

\begin{abstract}
   When metallic (or plasmonic) nanospheres are nearly touching, strong concentration of light can occur in the narrow gap regions. This phenomenon has a potential application in nanophotonics, biosensing and spectroscopy. The understanding of the strong interaction between the plasmonic spheres turns out to be quite challenging. Indeed, an extremely high computational cost is required to compute the electromagnetic field. Also, the classical method of image charges, which is effective for dielectric spheres system, is not valid for plasmonic spheres because of their negative permittivities. Here we develop new analytical and numerical methods for the plasmonic spheres system by clarifying the connection between transformation optics and the method of image charges. We derive fully analytic solutions valid for two plasmonic spheres. We then develop a hybrid numerical scheme for computing the field distribution produced by an arbitrary number of spheres. Our method is highly efficient and accurate even in the nearly touching case and is valid for plasmonic spheres.
\end{abstract}



\section{Introduction}

Controlling light at the nanoscale is a challenging problem. By using conventional optical devices, one cannot focus light into a spot smaller than a micron-sized region due to the diffraction limit. To overcome this fundamental difficulty, new optical materials are required. Recently, noble metal nanoparticles such as gold and silver have been extensively studied and utilized due to their unique optical properties. When visible light is incident, the electromagnetic fields near the surfaces of the particles show strongly resonant and oscillating behavior. In other words, they can strongly interact with light. This phenomenon is called the {plasmon resonance}. So metallic nanoparticles are often called  plasmonic nanoparticles. For practical applications of the plasmon resonance in nanophotonics, we refer to \cite{Plasmon_Rev_Gramotnev, PH, Schuller, Pen2, Pen3}. Roughly speaking, the plasmon resonance originates from the negative permittivity of metals. Contrary to plasmonic nanoparticles, ordinary dielectric nanoparticles with positive permittivity cannot strongly interact with light.

Among various plasmonic structures, the system of metallic spheres is of fundamental importance. When the spheres get close to touching, their interaction is so strong that plasmonic resonant fields can be greatly squeezed into the narrow gap region between them \cite{Pen, Rom, Sweatlock, Stockman_PRL, Nordlander1, Ory1}. Moreover, they can support collective resonance modes such as Fano resonances \cite{Fano, PH_fano}. These phenomena can have great impact on the design of nanophotonic devices, biosensing and spectroscopy \cite{Pen2,Pen3,Schuller, Plasmon_Rev_Gramotnev}. However, the problem of strong plasmonic interaction between nearly touching spheres is quite challenging to investigate both analytically and numerically.

We first discuss the analytic difficulty. There are two approaches for understanding the interaction between two spheres: (i) Transformation Optics (TO) and (ii) the method of image charges. TO is a design method for novel optical devices which control electromagnetic waves in an unprecedented way, including invisibility cloaks \cite{Pendry2006, TO_Uhlmann, Dolin}. Recently, the TO approach has been applied to analyze various singular plasmonic structures. It provides a novel physical insight into light harvesting \cite{Pen2, Pen3}. In particular, TO gives exact analytical solutions for 2D systems. But, the 3D case is more complicated. For two 3D plasmonic spheres, Pendry et al. \cite{Pen} derived a quasi-analytic solution using a TO inversion mapping which transforms two spheres into a concentric shell. But their solution is not fully analytic and still requires a numerical computation.

Next we consider the method of image charges. The principle of the image method is to find fictitious sources which generate the desired field. For two 2D dielectric cylinders, an exact image series solution and its asymptotic properties were derived \cite{McP, McPhedran_Milton, AKLLL}. See also \cite{ah, BLY-ARMA-09,BLL-ARMA-15,ET-ARMA-13,Gorb-MMS-16,LY3,Yun-SIAP-07}. Although the 3D case is more difficult, Poladian succeeded in deriving an approximate but fully analytical image series solution for two 3D dielectric spheres \cite{Pol_thesis,Pol,Pol2}. Unfortunately, the image series solution is not convergent when the permittivity is negative. So it cannot describe the plasmonic interaction between two spheres. Therefore, both TO and the image method cannot provide a complete analytical description valid for two plasmonic spheres.

We now discuss the numerical difficulty. Let us consider an arbitrary number of spheres. If the spheres are well separated, computing the field distribution can be efficiently done. However, as the spheres get closer, the required computational cost dramatically increases. In this case, the field in the narrow gap between the spheres becomes nearly singular. So the multipole expansion method requires a large number of spherical harmonics and the finite element method (or boundary element method) requires a very fine mesh in the gap. Moreover, the linear systems to be solved are ill-conditioned. So conventional numerical methods are time consuming or inaccurate for this extreme case. Although TO approach provides an efficient numerical scheme, it cannot be applied when the number of spheres is greater than two. For an arbitrary number of 2D dielectric cylinders, Cheng and Greengard developed a hybrid numerical scheme combining the multipole expansion and the method of image charges \cite{CG}. They used the image source series to capture the close-to-touching interaction. Their scheme is extremely efficient and accurate even if the spheres are nearly touching. Their scheme has been generalized to 3D perfect conducting spheres \cite{C} and 3D dielectric spheres \cite{Gan}. However, as already mentioned, the image series is not convergent when the permittivity of the sphere is negative. Hence their hybrid scheme cannot be used for plasmonic sphere clusters. In short, there are currently no efficient numerical schemes for nearly touching plasmonic spheres system.

The purpose of this paper is to solve all these analytical and numerical difficulties. Specifically, our goal is twofold: (i) to derive a fully analytical solution for two plasmonic spheres, (ii) to develop a hybrid numerical scheme for computing the field generated by an arbitrary number of plasmonic spheres which are nearly touching.

The key idea of our approach is to establish a connection between TO and the method of image charges, which is interesting in itself. Indeed, we find explicit formulas which can convert image source series into TO-type solutions. As already mentioned, the image series cannot describe the plasmonic interaction due to the non-convergence. So we convert Poladian's image series solution into a TO-type series by using our connection formula, resulting in fully analytical solutions valid for two plasmonic spheres. 
Next we modify Cheng and Greengard's hybrid numerical scheme by replacing the image series with our new analytic TO-type solutions. We also show extreme efficiency and accuracy of the resulting scheme by presenting several numerical examples. Our proposed scheme is a result of the interplay between three analytical approaches: TO, the image method, and the multipole expansion. We expect that our results will play a fundamental role in studying the plasmonic interaction between nanospheres.

\section{Problem formulation}
We consider a system of $N$  spheres $B_j,j=1,2,...,N$ where each individual sphere $B_j$ has  permittivity $\epsilon_B$ and radius $R$. We assume that the background has the permittivity $\epsilon_0=1$. Since the plasmonic nanospheres are much smaller than the wavelength of the visible light, we can adopt the quasi-static approximation for the electromagnetic fields. Then the electric field $\mathbf{E}$ is represented as $\mathbf{E}(\Br,t) = \Re \{-\nabla V(\Br)e^{i\omega t}\}$  where $V$ is the (quasi-static) electric potential and $\omega$ is the frequency. We also assume a uniform incident field with intensity $E_0$ is applied in the $z$-direction. Then the potential $V$ satisfies
\beq\label{eqn:potentialV}
\begin{cases}
\nabla \cdot(\ep \nabla V)=0, &\quad \mbox{in }\mathbb{R}^3,
\\
V(\Br)= - E_0 z + O(|\Br|^{-2}), &\quad \mbox{as } |\Br|\rightarrow\infty,
\end{cases}
\eeq
where $\ep$ is the permtittivity distribution which takes the value $\ep_B$ (or $\ep_0=1$) on each sphere (or on the backround), respectively. It can be shown that the above equation is equivalent to the following transmission problem:
\beq\label{eqn:potentialV_trans}
\begin{cases}
\ds \Delta V =0, &\quad \mbox{in }\mathbb{R}^3 \setminus (\cup_{j=1}^N \partial B_j),
\\[0.2em]
\ds V|_-=V|_+, &\quad \mbox{on }\partial B_j, \quad j=1,2,...,N,
\\[0.5em]
\ds \ep_B\frac{\partial V}{\partial \mathbf{n}}\Big|_- = \frac{\partial V}{\partial \mathbf{n}}\Big|_+
, &\quad \mbox{on }\partial B_j, \quad j=1,2,...,N,
\\[0.5em]
V(\Br)= - E_0 z + O(|\Br|^{-2}), &\quad \mbox{as } |\Br|\rightarrow\infty,
\end{cases}
\eeq
where $\mathbf{n}$ is an outward unit normal vector and the subscript $+$ (or $-$) means the limit from outside (or inside), respectively.

The permittivity $\ep_B$ of each metallic sphere depends on the frequency $\omega$. According to the Drude model, $\ep_B$  is modeled as
\beq\label{Drude_model}
\ep_B=\ep_B(\omega) = 1-\frac{\omega_p^2}{\omega(\omega+i\gamma)},
\eeq
where $\omega_p>0$ is the plasma frequency and $\gamma>0$ is a small damping parameter. It is clear that $\Re\{\ep_B\}<0 $ for $\omega<\omega_p$. Note that, as $\omega\rightarrow 0$, the permittivity $|\ep_B|$ goes to infinity. In this paper, we assume a silver nanoparticle and fit Palik's data \cite{Palik} for silver by adding a few Lorentz terms to \cref{Drude_model}. For specific values of the fitting parameters, see \cite{vander_Pen}.

\section{Plasmon resonance}
Here we briefly discuss the mathematical structure of plasmonic resonant fields.
Recently, a rigorous and general theory of plasmon resonances for nanoparticles has been developed using the spectral analysis of the Neumann-Poincar\'e (NP) operator in \cite{pierre,AK,matias_scalar,matias}.
We emphasize that the theory is valid for arbitrary shaped particles with smooth boundary. 
There, it was shown that there is a sequence $\{\ep_n\}_{n=1}^\infty$ of negative real permittivities such that $\ep_1 \leq \ep_2 \leq \cdots <-1$ and there exists a nontrivial solution $V_n$ to the problem \cref{eqn:potentialV_trans} with $\ep_B = \ep_n$ and $E_0=0$. Moreover, the solution $V$ to the original problem \cref{eqn:potentialV_trans} can be written in the following spectral form (it is slightly modified for our purpose):
\beq\label{eqn:V_spectral}
V(\Br) = -E_0 z + \sum_{n=1}^\infty \frac{c_n}{\ep_B-\ep_n} V_n(\Br),
\eeq
where the $c_n$ are some constant coefficients. Typically, the solution $V_n$ shows oscillating behavior along the boundary of the particles. These solutions $V_n$ are called the plasmon resonance modes. We also call $\ep_n$ as the plasmonic resonant permittivities. It is also worth mentioning that $\lambda_n := (\ep_n+1)/(2(\ep_n-1))$ are eigenvalues of the NP operator and $V_n$ is a single layer potential of the  eigenfunction associated to $\lambda_n$.

The formula \cref{eqn:V_spectral} clearly shows how the plasmon resonance occurs. As already mentioned, the real part of $\ep_B(\omega)$ can take any negative values from $0$ to $-\infty$ for $\omega<\omega_p$. So $\ep_B$ can be close to some resonant permittivity $\ep_n$ since $\ep_n<-1$. Then, in view of \cref{eqn:V_spectral}, we see that the associated plasmon resonance mode $V_n$ will be amplified. As a result, the nanoparticles strongly couple with the incident light. On the contrary, if the permittivity $\ep_B$ is positive, the plasmon resonance never occurs.

\begin{figure*}
\begin{center}
\includegraphics[height=5.3cm]{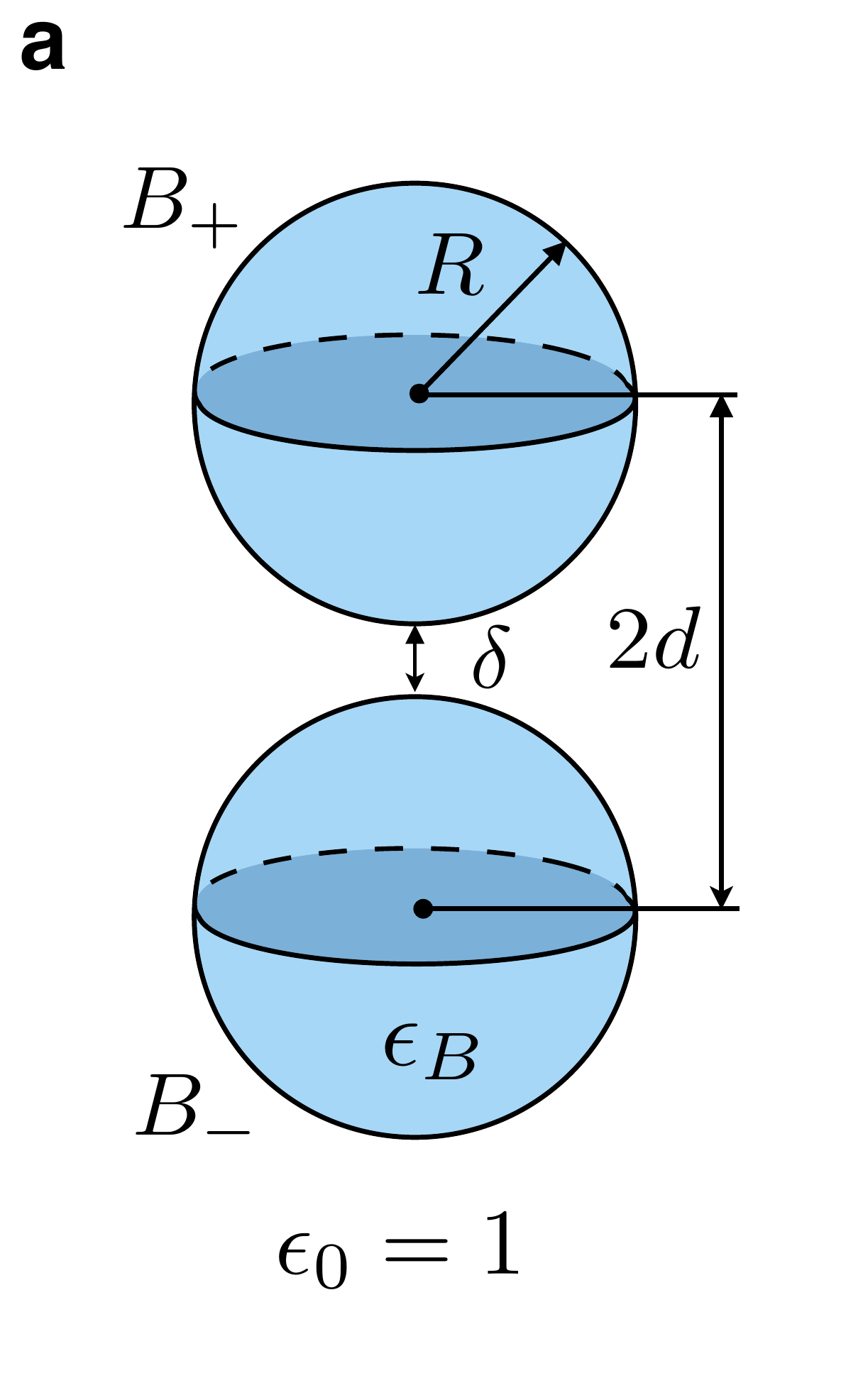}
\hskip 1cm
\includegraphics[height=5.3cm]{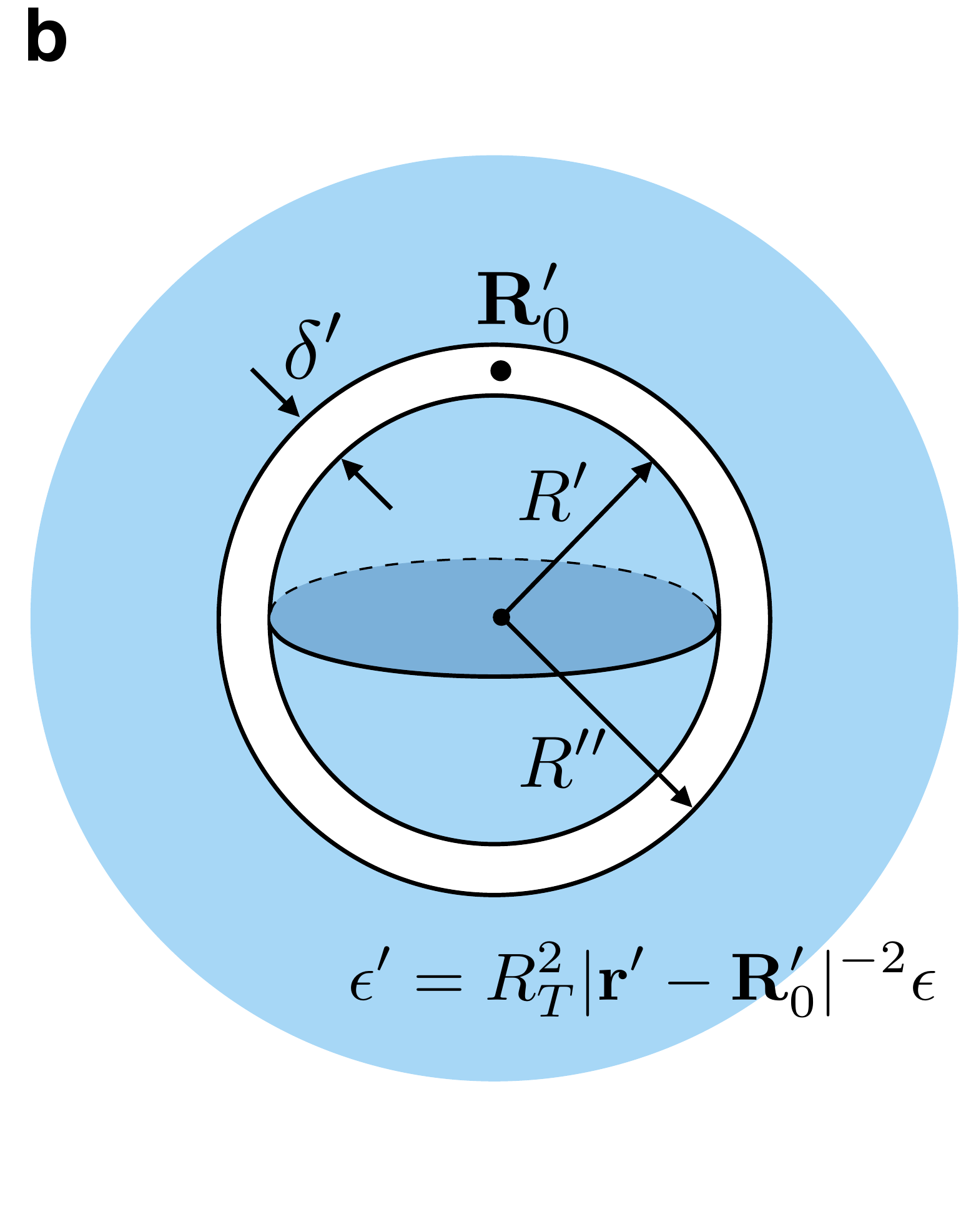}
\end{center}
\caption{Two spheres and the TO inversion mapping. (a) Two identical spheres, each of radius $R$ and the permittivity $\ep_B$, are separated by a distance  $\delta$. The distance between their centers is $2d$. The background permittivity is $\ep_0=1$. (b) The TO inverion mapping transforms the lower sphere $B_-$ (or the upper sphere $B_+$) into a sphere of radius $R'$ (or a hollow sphere of radius $R''$) centered at the origin, respectively.}
\label{fig1}
\end{figure*}

\section{Transformation Optics approach for two spheres}

Here we briefly review the TO approach by Pendry et al. \cite{Pen}. Let us assume there are only two spheres, {\it i.e.}, $N=2$. We first need to fix notations. Suppose that two spheres $B_1$ and $B_2$ are centered at $(0,0,+d)$ and $(0,0,-d)$, respectively. For convenience, let us denote $B_1=B_+$ and $B_2=B_-$. We also let $\delta$ be the gap distance between the two spheres. See Figure \ref{fig1}a.

To transform two spheres into a concentric shell, Pendry et al. \cite{Pen} introduced the inversion transformation $\Phi$ defined as
\begin{equation*}
\Br'=\Phi(\Br)= R_T^2 ({\Br-\mathbf{R}_0})/{|\Br-\mathbf{R}_0|^2} +\mathbf{R}_0',
\nonumber
\end{equation*}
where $\mathbf{R}_0, \mathbf{R}_0'$ and $R_T$ are given parameters (for the details, see Appendix \ref{sect1}). Then the transformed potential $V':=V\circ \Phi^{-1}$ satisfies $\nabla \cdot (\ep' \nabla V')=0$ where $\ep'$ is the transformed permittivity distribution defined by $\ep'(\Br') = R_T^2|\Br'-\mathbf{R}_0'|^{-2}\ep$.  Then, by taking advantage of the symmetry of the shell, they represented the potential $V$ in terms of the following basis functions:
\begin{equation*}
\mathcal{M}_{n,\pm}^m(\Br) = |\Br'-\mathbf{R}_0'|(r')^{\pm (n+\frac{1}{2})-\frac{1}{2}} Y_{n}^m(\theta',\phi'),
\end{equation*}
where $Y_{n}^m$ are the spherical harmonics. We will call $\mathcal{M}_{n,\pm}^m$ a TO basis.

Then the potential $V$ outside the two spheres in a uniform field $(0,0,E_0)$ can be represented as follows:
\beq\label{TO_V_eqn_main}
V(\Br) =  - E_0 z + \sum_{n=0}^\infty A_n \big(\mathcal{M}_{n,+}^0(\Br) -  \mathcal{M}_{n,-}^0(\Br)\big).
\eeq
Here, the coefficients $A_n$ can be determined by solving some tridiagonal system or recurrence relations (see Appendix \ref{app_sec:recur_An} for the details). Unfortunately, it cannot be solved analytically and a numerical computation is required to get the $A_n$. It is worth mentioning that Goyette and Navon \cite{GN} derived a similar solution using bispherical coordinates.

We will derive an approximate analytical expression for $A_n$ by establishing the explicit connection between TO and the method of image charges. We shall also see that our approximate expression captures the singular nature of the close-to-touching interaction completely.

\begin{figure*}
\begin{center}
\includegraphics[height=3.8cm]{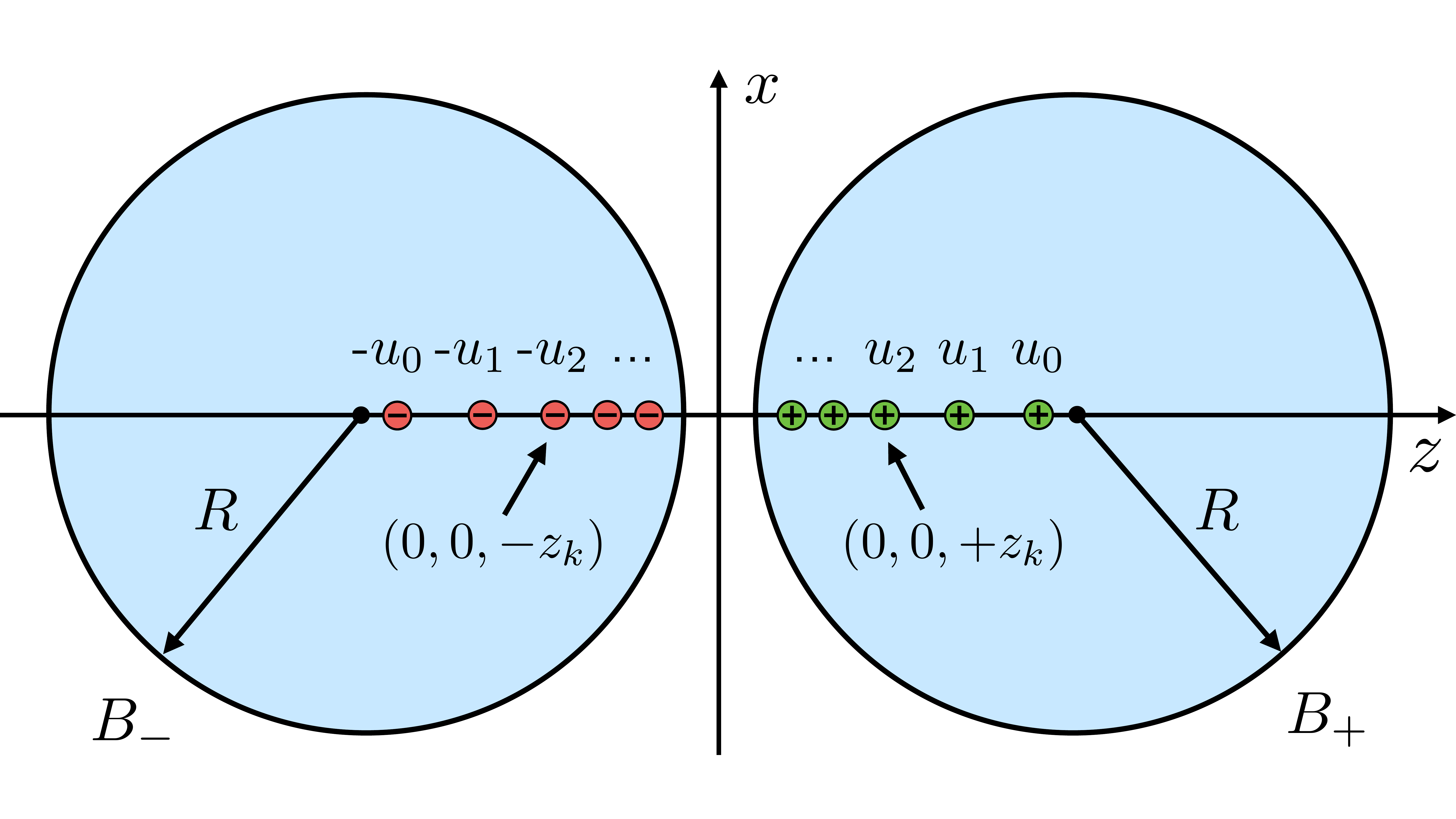}
\end{center}
\caption{Image charges for two spheres. Red and green circles represent image charges placed along the $z$-axis.}
\label{figure21}
\end{figure*}

\section{Method of image charges for two spheres}
Now we discuss the method of images. In the case of two 2D circular cylinders, the exact solution is represented as an infinite series of image point sources. McPhedran, Poladian and Milton \cite{McP} derived its asymptotic properties in the nearly touching case by approximating the sequence of image charges or dipoles. However, for two 3D dielectric spheres, an exact solution cannot be obtained due to the appearance of a continuous line image source \cite{Lindell1, Lindell2, Neumann}.  Poladian observed that the continuous source can be well approximated by a point charge and then the imaging rule becomes similar to the 2D case. He then derived an approximate but analytic image series solution and its asymptotic properties \cite{Pol_thesis,Pol,Pol2}. See also \cite{BLY-ARMA-09,Gorb-MMS-16,KLY-MA-15,KLY-SIAP-14,Lekner-JE-11,Lekner-PRSA-12,LY-CPDE-09,Yun-arXiv}.

Let us briefly state Poladian's solution for two 3D dielectric spheres (for the reader's convenience, we include the details of Poladian's image method in Appendix \ref{sec:Poladian}). Let $\tau = ({\ep_B-1})/({\ep_B+1})$, $s=\cosh^{-1}(d/R)$ and $\alpha=R\sinh s$. Suppose that two point charges of strength $\pm 1$ are located at  $(0,0,\pm  z_0)\in B_\pm$, respectively. By Poladian's imaging rule, they produce an infinite series of image charges of strength $\pm u_k$ at $(0,0,\pm z_k)$ for $k=0,1,2,\cdots$, where $z_k$ and $u_k$ are given by
\beq\label{zkuk}
{z_k} =  \alpha  \coth (ks+s+t_0),\quad u_k = \tau^{k}  \frac{\sinh (s+t_0)}{\sinh (ks+s+t_0)}.
\eeq
Here, the parameter $t_0$ is such that $z_0 =\alpha \coth (s+t_0)$. See Figure \ref{figure21}. The potential $U(\Br)$ generated by all the above image charges is given by
\beq\label{image_series_U}
U(\Br) = \sum_{k=0}^\infty u_k ( G(\Br-\mathbf{z}_k) - G(\Br+\mathbf{z}_k) ),
\eeq
where $\mathbf{z}_k=(0,0,z_k)$ and $G(\Br) = {1}/({4\pi|\Br|})$ is the potential generated by a unit point charge located at the origin.

Let us turn to the solution $V$ to the problem \cref{eqn:potentialV_trans}, which is the potential generated by the two spheres $B_+\cup B_-$ under a uniform incident field $(0,0,E_0)$. Let $p_0$ be the induced polarizability when a single sphere is subjected to the uniform incident field, that is, $p_0 =  E_0 R^3 {2\tau }/({3-\tau})$. Using the potential $U(\Br)$, we can represent the approximate solution for $V(\Br)$ as follows (see Appendix \ref{app_Poladian_uniform} for its derivation): for $|\tau|\approx 1$, we have
\beq\label{image_approximate_final}
V(\Br)\approx -E_0 z + 4\pi p_0 \frac{\partial  (U(\Br))}{\partial z_0}\Big|_{z_0=d}+ Q U(\Br)|_{z_0=d},
\eeq
where $Q$ is a constant chosen so that the right-hand side in equation (\ref{image_approximate_final}) has no net flux on the surface of each sphere. The accuracy of the approximate formula \cref{image_approximate_final} improves as {$|\ep_B|$} increases and it becomes exact when {$|\ep_B|=\infty$}. Moreover, its accuracy is pretty good even if the value of {$|\ep_B|$} is moderate.

We now explain the difficulty in applying the the image series solution \cref{image_approximate_final} to the plasmonic spheres. In view of the expressions \cref{zkuk} for $u_k$, we can see that the image series solution \cref{image_approximate_final} is not convergent when $|\tau|>e^s$. For plasmonic materials such as gold and silver, the real part of the permittivity $\ep_B$ is negative over optical frequencies and then the corresponding parameter $|\tau|$ can attain any value in the interval $(e^s,\infty)$. Moreover, it turns out that all the plasmonic resonant values for $\tau$ are contained in the set $\{\tau\in\mathbb{C}:|\tau|>e^s\}$. So, the image method solution \cref{image_approximate_final} cannot describe the plasmonic interaction between the spheres due to the non-convergence.

\section{Connection formula from image charges to TO}

Now we clarify the connection between TO and the method of image charges. We derive an explicit formula which converts an image charge to TO-type solutions as shown in the following lemma (see Appendix \ref{app_single_image_TO} for its proof).
\begin{lemma}\label{lem_connect_charge_TO}
(Converting an image charge to TO) The potential $u_k G(\Br\mp \mathbf{z}_k)$ generated by the image charge at $\pm\mathbf{z}_k$ can be rewritten using the TO basis as follows:
for $\Br\in\mathbb{R}^3\setminus(B_+ \cup B_-)$,
\beq\label{connect_charge_TO}
u_k G(\Br\mp\mathbf{z}_k) = \frac{\sinh(s+t_0)}{4\pi \alpha}\sum_{n=0}^\infty \big[\tau e^{-(2n+1)s}\big]^k 
e^{-(2n+1)(s+t_0)}
\mathcal{M}^0_{n,\pm}(\Br).
\eeq
\end{lemma}
This identity plays a key role in our derivation of the approximate analytical  solution. As mentioned previously, the reason why the image charge series \cref{image_series_U} does not work for plasmonic spheres is because the factor $(\tau e^{-s})^k$ may not converge to zero as $k\rightarrow \infty$. But the above connection formula helps us overcome this difficulty. If we sum up all the image charges in equation \cref{connect_charge_TO}, we can see that the summation over $k$ can be evaluated analytically using the following identity:
\begin{equation*}
\sum_{k=0}^\infty \big[\tau e^{-(2n+1)s}\big]^k  = \frac{e^{(2n+1)s}}{e^{(2n+1)s}-\tau}.
\end{equation*}
Therefore, from \cref{image_series_U} and Lemma \ref{lem_connect_charge_TO}, we obtain the following result.

\begin{theorem}\label{thm_image_to_TO_U} (Converting image charge series to TO)
Let $U(\Br)$ be the image charge series defined as in \cref{image_series_U}. Then it can be rewritten using TO basis as follows: for $\Br\in\mathbb{R}^3\setminus(B_+ \cup B_-)$,
\begin{equation} \label{rhs}
U(\Br) = \frac{\sinh(s+t_0)}{4\pi \alpha}  \sum_{n=0}^\infty \frac{e^{-(2n+1)t_0}}{e^{(2n+1)s}-\tau}
\Big(\mathcal{M}_{n,+}^0(\Br) - \mathcal{M}_{n,-}^0(\Br) \Big). 
\end{equation}
\end{theorem}
Clearly, the right-hand side of (\ref{rhs}) does converge for any $|\tau|>e^s$ provided that $\tau\neq e^{(2n+1)s}$.

\section{Analytical solution for two plasmonic spheres}
Here we derive an analytic approximate solution $V$ for two plasmonic spheres in a uniform incident field $(0,0,E_0)$. Moreover, we shall see that our analytical approximation completely captures the singular behavior of the exact solution. This feature will be essentially used to develop our hybrid numerical scheme. We only consider the case when the incident field is in the direction of the $z$-axis. In the case of the $x$ or $y$-axis, a high field concentration  in the gap does not happen \cite{Pen, Rom}.

To derive the solution valid for two plasmonic spheres, we convert the image series \cref{image_approximate_final} into a TO-type solution by using the connection formula \cref{rhs}. The result is shown in the following theorem (see Appendix \ref{app_pf_main_thm1} for its proof).

\begin{theorem}\label{main_thm1}
 If {$|\tau|\approx 1$}, the following approximation for the electric potential $V(\Br)$ holds: for $\Br\in \mathbb{R}^3\setminus(B_+\cup B_-)$, 
\beq\label{approximate_final}
V(\Br)\approx-E_0 z +    \sum_{n=0}^\infty  \widetilde{A}_n \Big(\mathcal{M}_{n,+}^0(\Br)-\mathcal{M}_{n,-}^0(\Br) \Big),
\eeq
where the coefficient $\widetilde{A}_n$ is given by
\begin{align*}
\ds\widetilde{A}_n &= E_0\frac{2\tau  \alpha}{3-\tau}\cdot\frac{2n+1- K_0}{e^{(2n+1)s}-\tau},
\\
\ds K_0 &= \sum_{n=0}^\infty \frac{2n+1}{e^{(2n+1)s}-\tau}\bigg/ \sum_{n=0}^\infty \frac{1}{e^{(2n+1)s}-\tau}.
\end{align*}
\end{theorem}

As expected, the above approximate expression is valid even if $|\tau|>e^s$. Therefore, it can furnish useful information about the plasmonic interaction between the two spheres. As a first demonstration, let us investigate the (approximate) resonance condition, that is, the condition for $\tau$ at which the coefficients $\widetilde{A}_n$ diverge. One might conclude that the resonance condition is given by $\tau=e^{(2n+1)s}$. However, one can see that $\widetilde{A}_n$ has a removable singularity at each $\tau=e^{(2n+1)s}$. In fact, the (approximate) resonance condition turns out to be
\beq\label{resonance_cond}
\sum_{n=0}^\infty \frac{1}{e^{(2n+1)s}-\tau} =0.
\eeq
In other words, the plasmon resonance does happen when $\tau$ is one of zeros of equation \cref{resonance_cond}. It turns out that the zeros $\{\tau_n\}_{n=0}^\infty$ lie on the positive real axis and satisfy, for $n=0,1,2,\cdots$,
\beq\label{estim_tau_n}
e^{(2n+1)s}< \tau_n < e^{(2n+3)s}. 
\eeq
The above estimate help us understand the asymptotic behavior of the resonance when two spheres get closer. As the gap distance $\delta$ goes to zero, the parameter $s$ also goes to zero (in fact, $s=O(\delta^{1/2})$). Then, in view of \cref{estim_tau_n}, $\tau_n$ will converge to $1$ and the corresponding permittivity $|\ep_n|$ goes to infinity. Also, the corresponding frequency $\omega_n$ goes to zero according to Drude's model. This phenomenon is sometimes called the red-shift of the (bright) resonance modes \cite{Rom, Ory1}. Since our approximate analytical formula \cref{approximate_final} for $V$ becomes more accurate as $|\ep_B|$ increases, we can expect that accuracy of the plasmonic resonant field improves as the separation distance goes to zero. It also indicates that our formula captures the singular nature of the field distribution completely. Also, the difference between $\tau_n$ and $\tau_{n+1}$ decreases, which means that the spectrum becomes a nearly continuous one. It is worth  mentioning that $1/(2\tau_n)$ gives the approximate eigenvalues of the Neumann-Poincar\'e operator for two spheres.

\begin{figure*}
\begin{center}
\includegraphics[height=4.5cm]{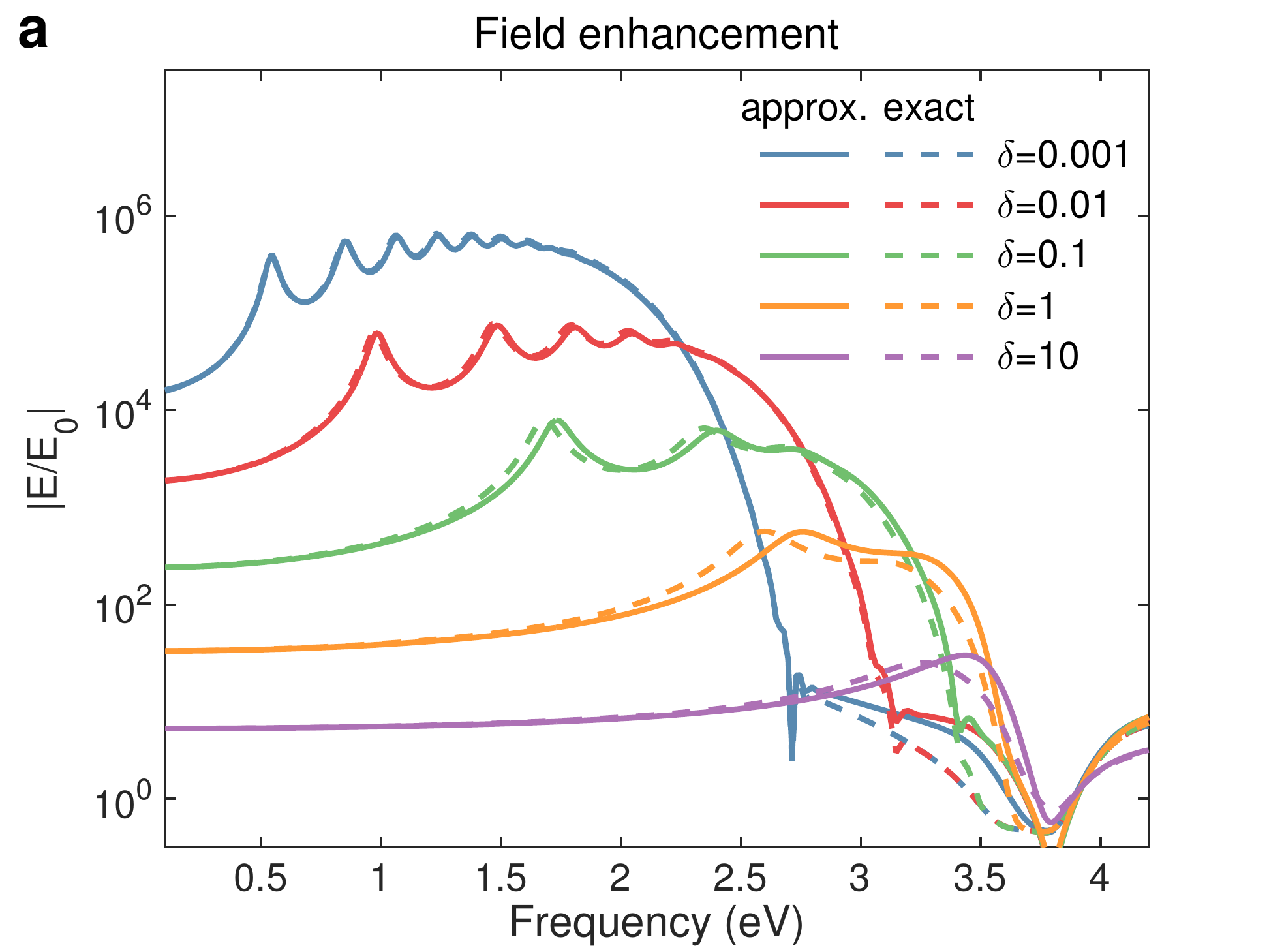}
\includegraphics[height=4.5cm]{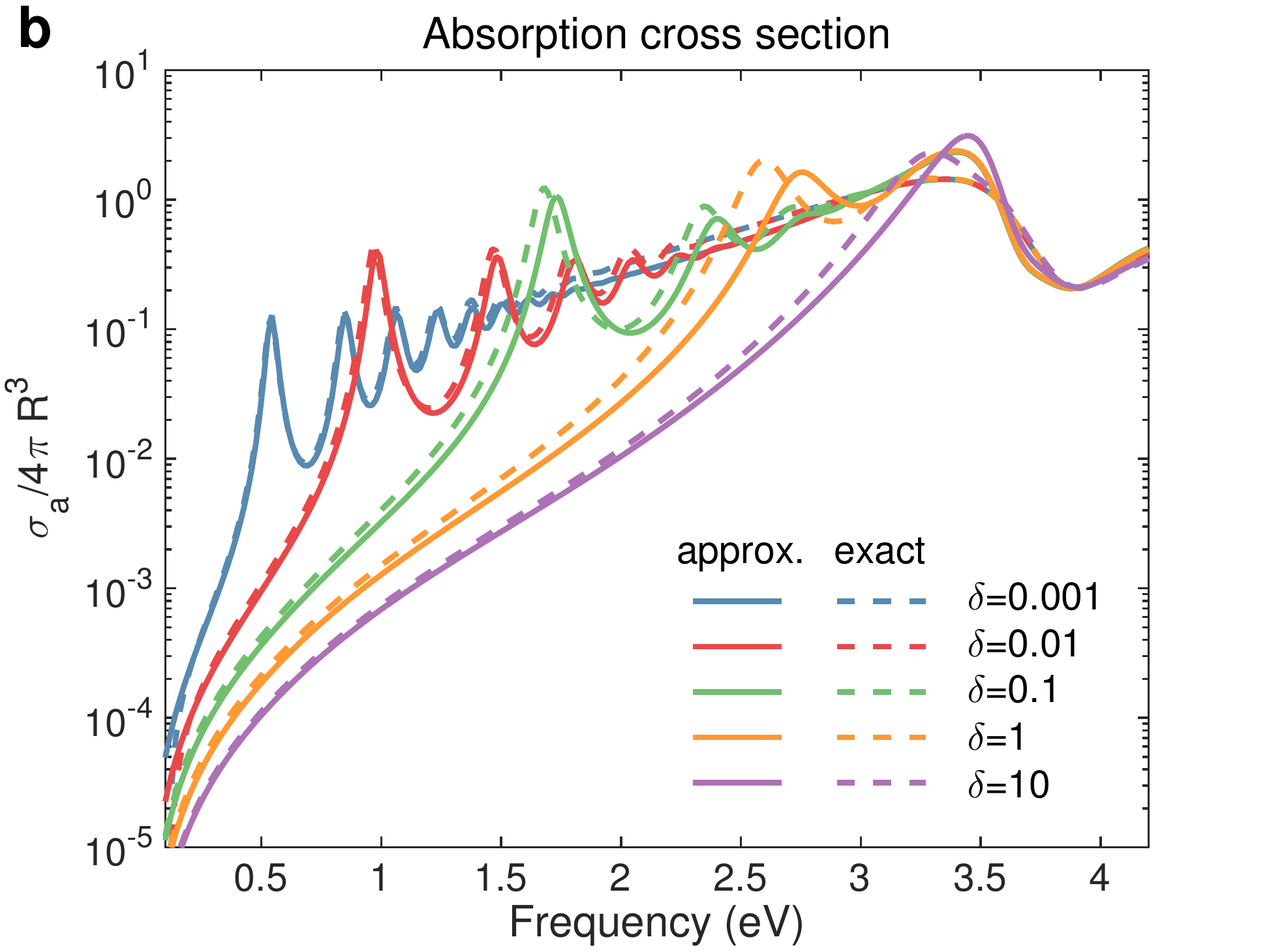}
\end{center}
\caption{Exact solution vs Analytic approximation. (a) Field enhancement plot as a function of the frequency $\omega$ for various separation distances $\delta$.
  The solid lines represent the approximate analytical solution and the dashed lines represent the exact solution.   Two identical silver spheres of radius $30$ nm are considered. 
(b) Same as (a) but for the absorption cross section.}
\label{fig3} 
\end{figure*}

We now derive approximate formulas for the field  at the gap center and for the absorption cross section. From Theorem \ref{main_thm1}, we obtain the following approximation (see Appendix \ref{app_field_gap_absorb} for the details):
\begin{align*}
\ds E(0,0,0) &\approx E_0 - E_0\frac{8\tau}{3-\tau}\bigg[\sum_{n=0}^\infty \frac{(2n+1)^2}{e^{(2n+1)s}-\tau}(-1)^n \\ 
&\quad\ds - K_0\sum_{n=0}^\infty \frac{2n+1}{e^{(2n+1)s}-\tau}(-1)^n\bigg].
\end{align*}
In the quasi-static approximation, the absorption cross section $\sigma_{a}$ is defined by
$\sigma_{a} = \omega \mbox{Im}\{ {p}\}, 
$
where $p$ is the polarizability of the system of two spheres. From Theorem \ref{main_thm1}, $\sigma_a$ is approximated as follows (see again Appendix \ref{app_field_gap_absorb}):
\begin{align*}
\ds \sigma_{a}  &\approx  \omega  E_0\frac{8\tau  \alpha^3}{3-\tau}\bigg[\sum_{n=0}^\infty \frac{(2n+1)^2}{e^{(2n+1)s}-\tau} 
\\ 
\ds &\quad -\bigg(\sum_{n=0}^\infty \frac{2n+1}{e^{(2n+1)s}-\tau}\bigg)^2\bigg/ \sum_{n=0}^\infty \frac{1}{e^{(2n+1)s}-\tau}\bigg].
\end{align*}

\begin{figure*}
\begin{center}
\includegraphics[height=3.8cm]{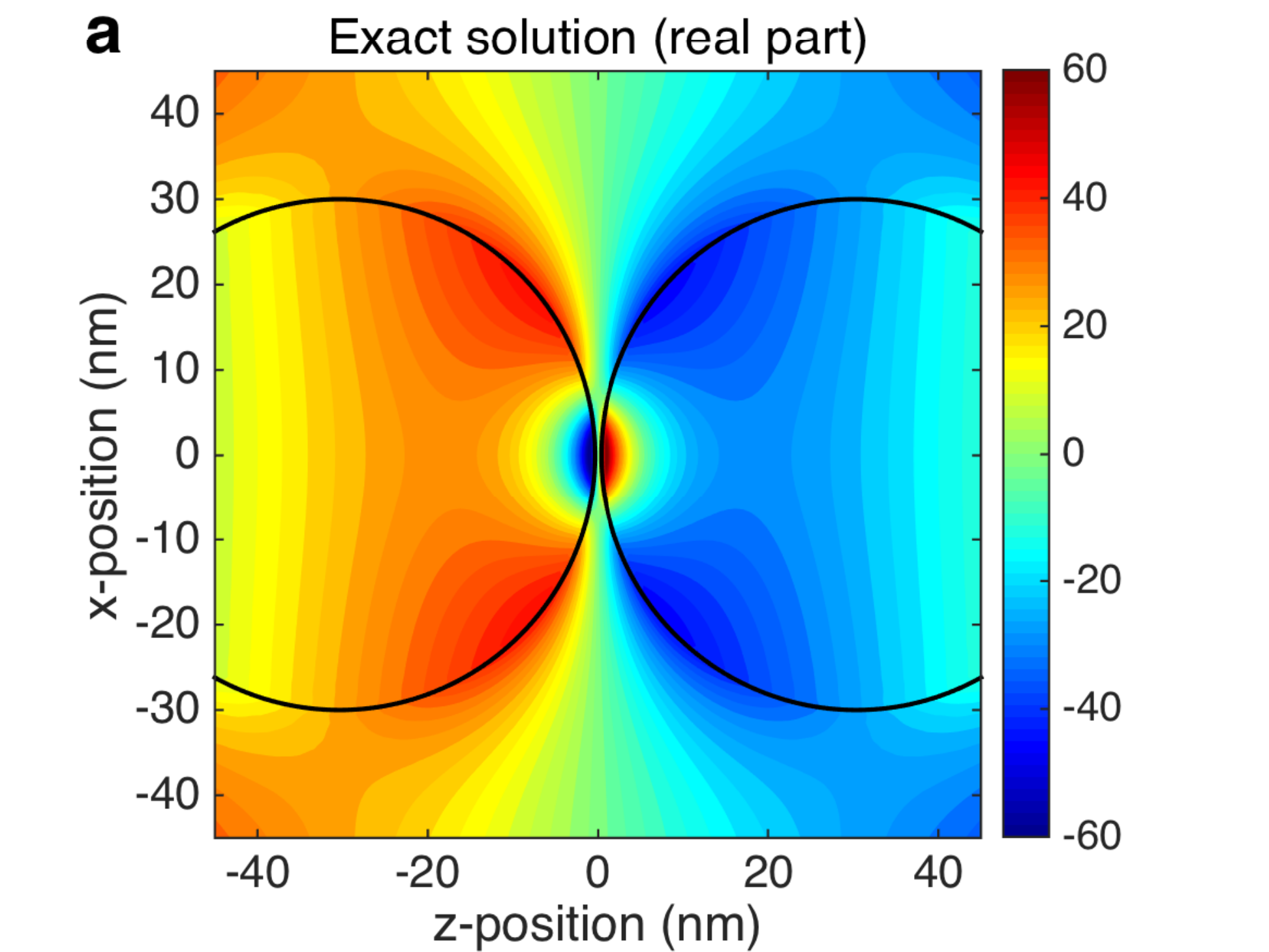}
\includegraphics[height=3.8cm]{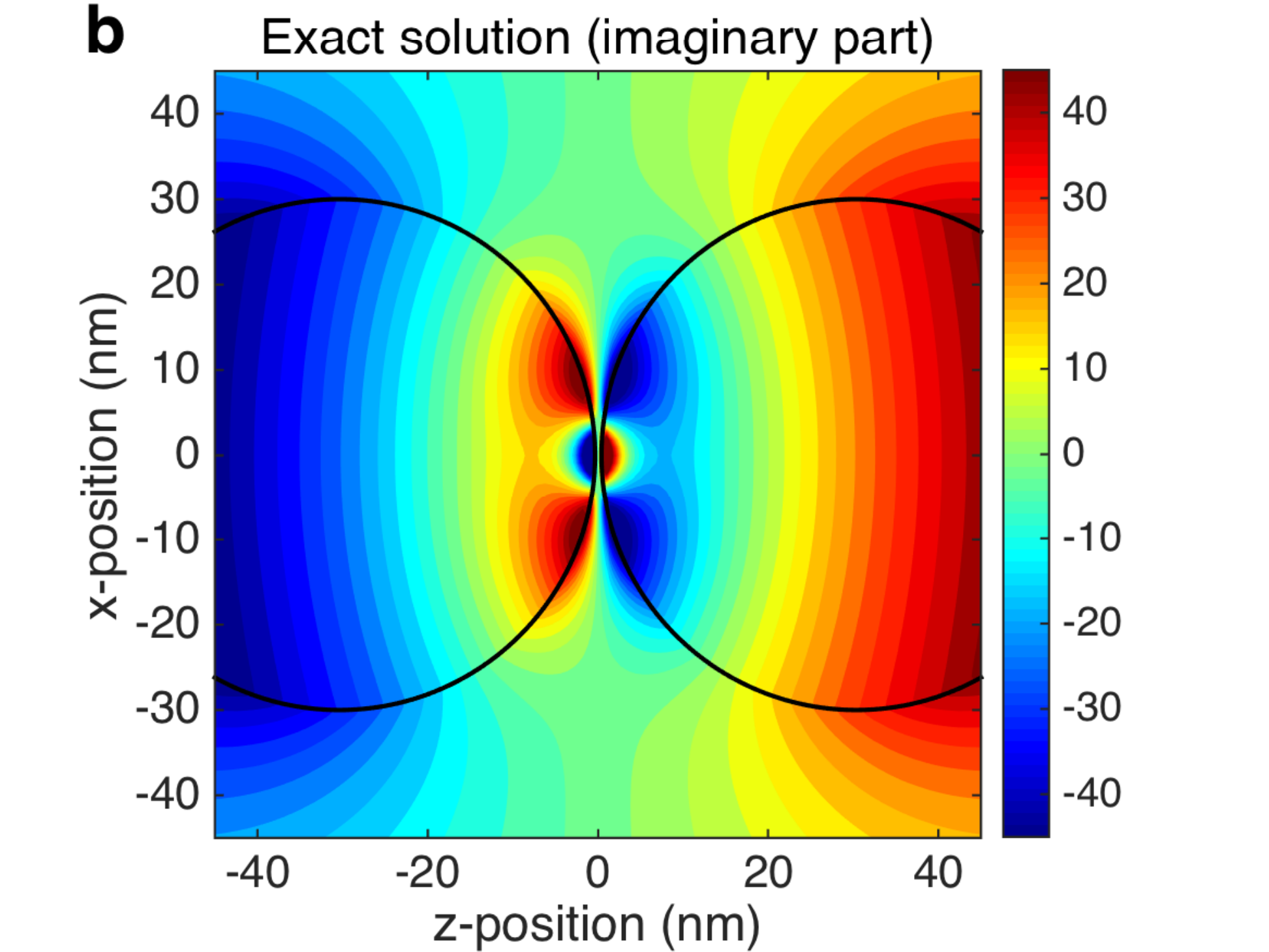}
\vskip.2cm
\includegraphics[height=3.8cm]{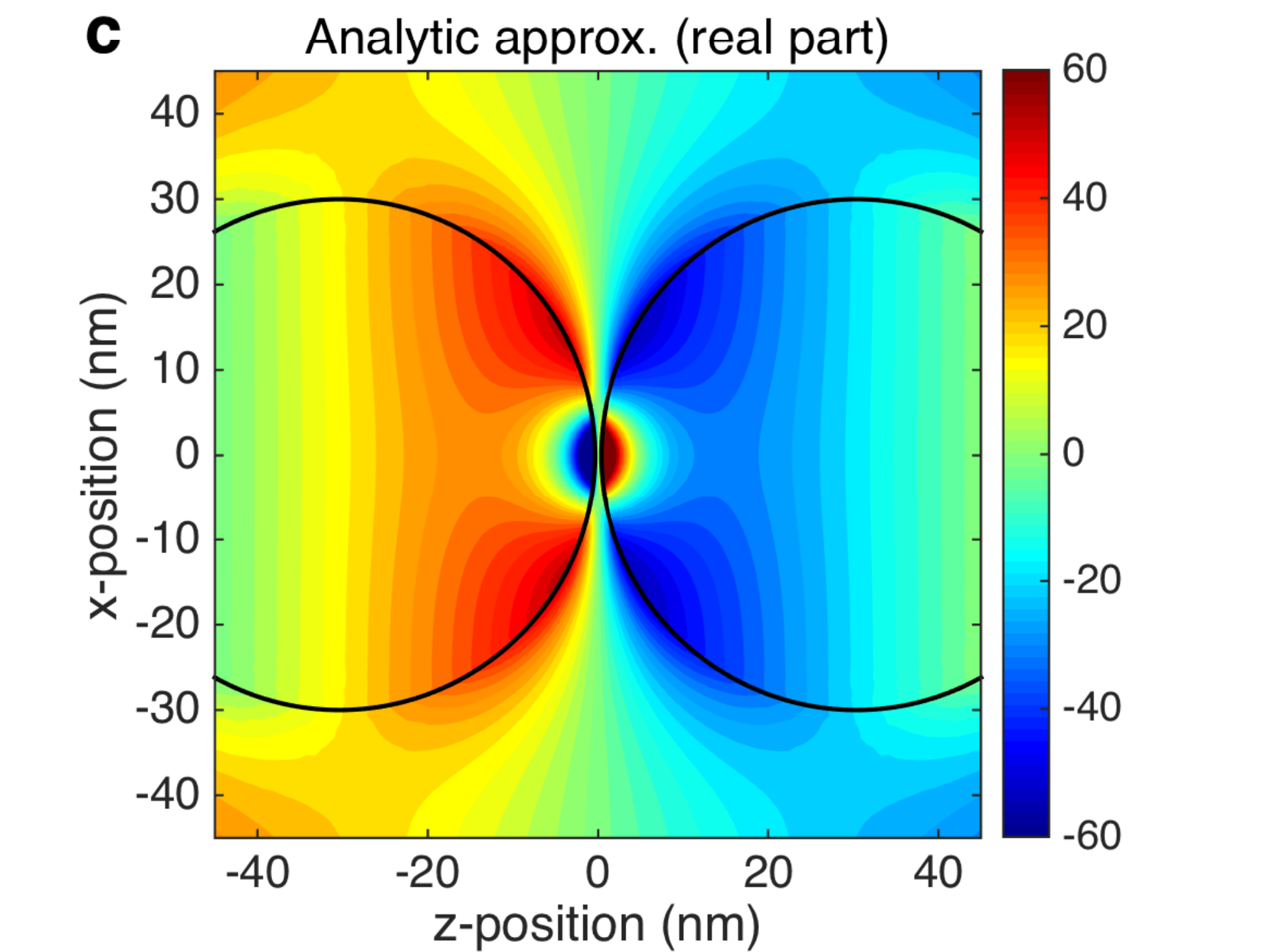}
\includegraphics[height=3.8cm]{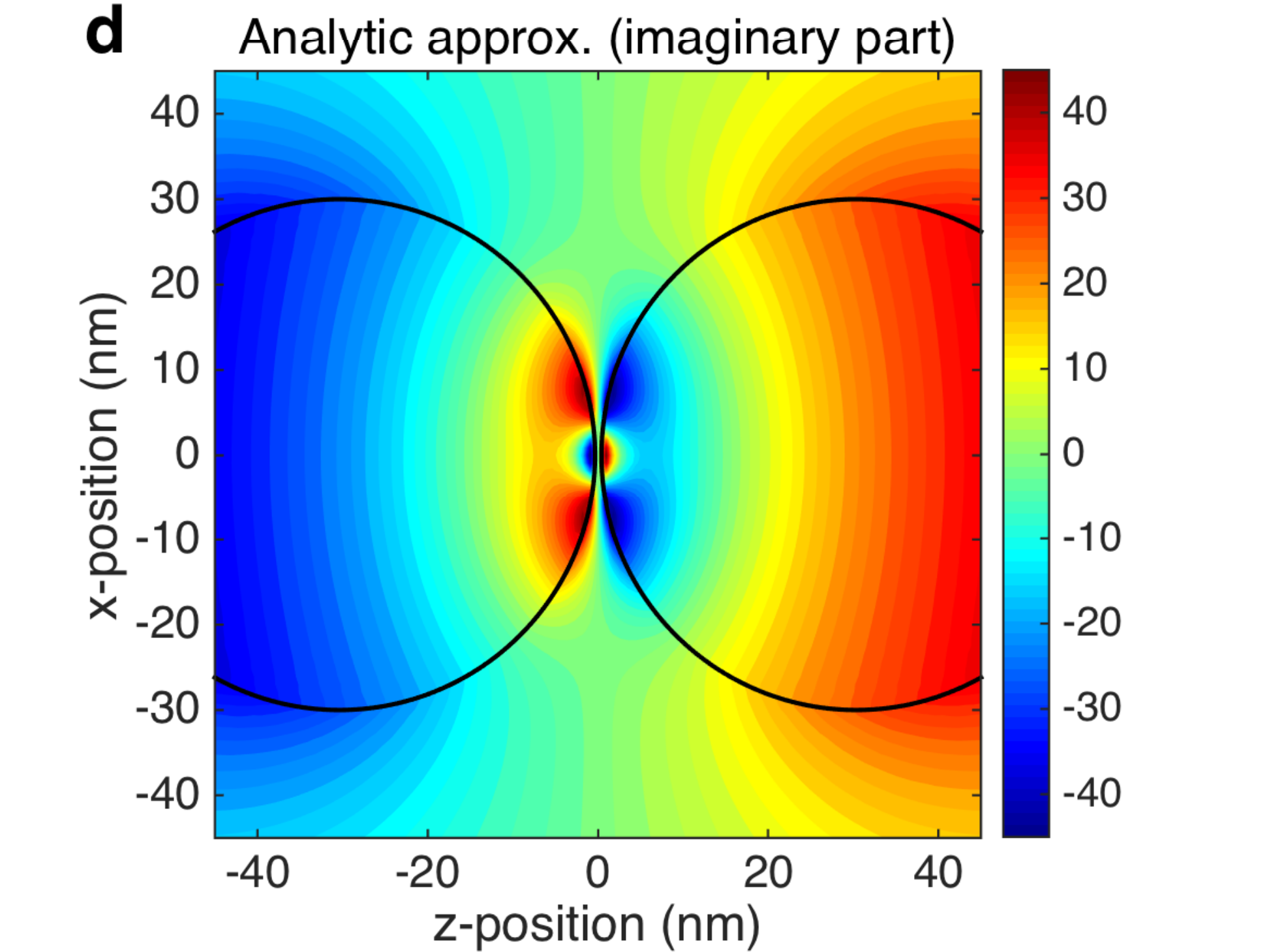}
\end{center}
\caption{Potential distributions for two identical silver spheres of radius $30$ nm separated by $\delta=0.25$ nm.  (a,b) Real and imaginary parts of the exact solution.   The uniform incident field $(0,0,\mbox{Re}\{ e^{i\omega t} \})$ is applied at the frequency $\omega=3.0$ eV in $z$-direction. (c,d) Same as (a,b) but for the  analytical approximate solution.}
\label{fig4}
\end{figure*}

We compare the above approximate formulas with the exact ones. Figure \ref{fig3} represents respectively the field enhancement and the absorption cross section $\sigma_{a}$ as functions of the frequency $\omega$ for various distances ranging from $0.001$ nm to $10$ nm. The good accuracy of our approximate formulas over broad ranges of frequencies and gap distances is clearly shown. As mentioned previously, the accuracy improves as the spheres get closer. The red-shift of the plasmon resonance modes is also shown. It is worth to mention that Schnitzer \cite{Ory1} performed an asymptotic analysis for  the field enhancement, the polairizability and their red-shift behavior. In Figure \ref{fig4}, we compare the exact and approximate electric potential distributions. They are also in good agreement and the field concentration in the gap region is observed.

\section{Hybrid numerical scheme for many-spheres system}

Now we consider a system of an arbitrary number of plasmonic spheres. If all the spheres are well separated, then the multipole expansion method is efficient and accurate for computing the field distribution (see Appendix \ref{app_multipole}). But, when the spheres are close to each other, the problem becomes very challenging  since the charge densities on each sphere are nearly singular. To overcome this difficulty, Cheng and Greengard developed a hybrid numerical scheme combining the multipole expansion and the method of images \cite{CG}. See also \cite{C,Gan}.

Let us briefly explain the main idea of Cheng and Greengard's method. In the standard multipole expansion method, the potential is represented as a sum of general multipole sources $\mathcal{Y}_{lm}(\Br) = {Y_{l}^m(\theta,\phi)}/{r^{l+1}}$ located at the center of each of spheres. Suppose that a pair of spheres is close to touching. For convenience, let us identify the pair as $B_+\cup B_-$.  A multipole source $\mathcal{Y}_{lm}$ located at the center of $B_+$ generates an infinite sequence of image multipole sources by Poladian's imaging rule. Let us denote the resulting image multipole potential by $U_{lm}^+$. We also define $U_{lm}^-$ in a similar way. The detailed image series representation for $U_{lm}^\pm$ can be found in Appendix \ref{app_Poladian_multipole}. Roughly speaking, Cheng and Greengard modified the multipole expansion by replacing a multipole source $\mathcal{Y}_{lm}$ with its corresponding image multipole series $U_{lm}^\pm$.

Since the image series $U_{lm}^\pm$ captures the close-to-touching interactions analytically, their scheme is extremely efficient and highly accurate even if the distance between the spheres is extremely small. However, the image mulipole series $U_{lm}^\pm$ are not convergent for $|\tau|>e^s$.
Hence it cannot be applied to plasmonic spheres clusters. Therefore, for extending Cheng and Greengard's method to the plasmonic case, it is essential to establish an explicit connection between the image multipole series $U_{lm}^\pm$ and TO. We derive the following formula for this connection (see Appendix \ref{app_pf_thm_multipole_TO} for its proof).

\begin{theorem}\label{thm_multipole_TO}
(Converting image multipole series to TO) Assume that the integers $l$ and $m$ are such that $l\geq 1$ and $-l\leq m\leq l$. The potential $U_{lm}^\pm$ can be rewritten in terms of TO basis as follows: for $\Br\in \mathbb{R}^3\setminus(B_+\cup B_-)$,
\begin{align}
\ds U^\pm_{lm}(\Br) &= \sum_{n=|m|}^\infty   \frac{ g_{n}^m \mathcal{D}_{lm}^\pm[\lambda_{n}^m] }{ e^{2(2n+1)s}-\tau^2 }(e^{(2n+1)s}\mathcal{M}_{n,\pm}^m(\Br)-\tau\mathcal{M}_{n,\mp}^m(\Br))\nonumber
\\
\ds &\quad 
- \delta_{0m}\frac{\widetilde{Q}^\pm_{l,1}}{2} \sum_{n=0}^\infty \frac{\mathcal{M}_{n,+}^0(\Br)+(-1)^l\mathcal{M}_{n,-}^0(\Br)}{ e^{(2n+1)s}+(-1)^l\tau }
\\
\ds &\quad 
\mp \delta_{0m}\frac{\widetilde{Q}^\pm_{l,2}}{2} \sum_{n=0}^\infty \frac{\mathcal{M}_{n,+}^0(\Br)-(-1)^l\mathcal{M}_{n,-}^0(\Br)}{ e^{(2n+1)s}-(-1)^l\tau },\nonumber
\end{align}
where $g_n^m,\lambda_n^m$ and ${Q}_l^\pm$ are given by
\begin{align}
\ds g_{n}^m&=\frac{1}{\alpha^{|m|+1}}\frac{2^{|m|}}{\sqrt{(2|m|)!}} \sqrt{\frac{(n+|m|)!}{(n-|m|)!}},
\nonumber
\\
\ds \lambda_{n}^m &=  [\sinh(s+t_0)]^{2|m|+1} \,e^{-(2n+1)t_0},
\label{def_gnm_lambdanm_Q}
\\
\ds \widetilde{Q}_{l,i}^\pm &=  \sum_{n=0}^\infty  \frac{ (\pm 1)^{l} g_{n}^0 \mathcal{D}_{l0}^\pm[\lambda_{n}^0] }{ e^{(2n+1)s} - (-1)^{l+i}\tau }\Bigg/\sum_{n=0}^\infty \frac{1}{ e^{(2n+1)s}-(-1)^{l+i}\tau }.
\nonumber
\end{align}
Here, $\mathcal{D}_{lm}^\pm[\cdot]$ is defined by \cref{eqn:D_lm_pm} and $\delta_{lm}$ is the Kronecker's delta function.
\end{theorem}

Clearly, the above TO representation for $U_{lm}^\pm$ does converge for $|\tau|>e^s$. Based on these analytic formulas, we develop a modified hybrid numerical scheme for the plasmonic spheres system. Specifically, we modify Cheng and Greengard's hybrid scheme by replacing the image multipole series $U_{lm}^\pm$ with its TO version using Theorem \ref{thm_multipole_TO}. The resulting hybrid scheme is valid for plasmonic spheres. Our new analytic TO-type solutions for $U_{lm}^\pm$ capture the singular behavior of the field distribution in the gap regions. So our modified hybrid scheme is extremely efficient and accurate even when the spheres are nearly touching. For a detailed description of the proposed scheme, we refer to Appendix \ref{sec:hybrid}.

\begin{figure*}
\begin{center}
\includegraphics[height=3.7cm]{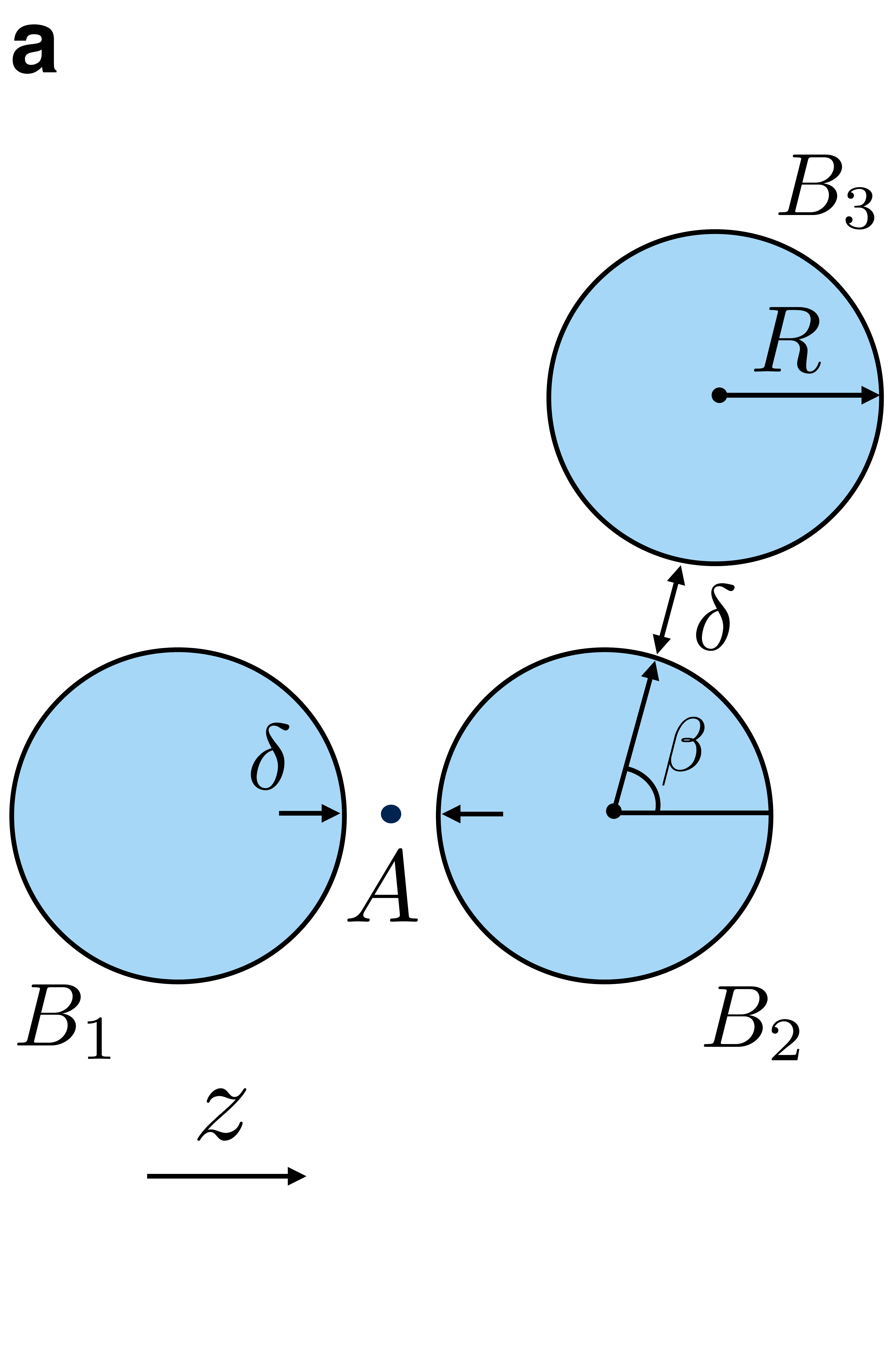}
\hskip.2cm
\includegraphics[height=3.7cm]{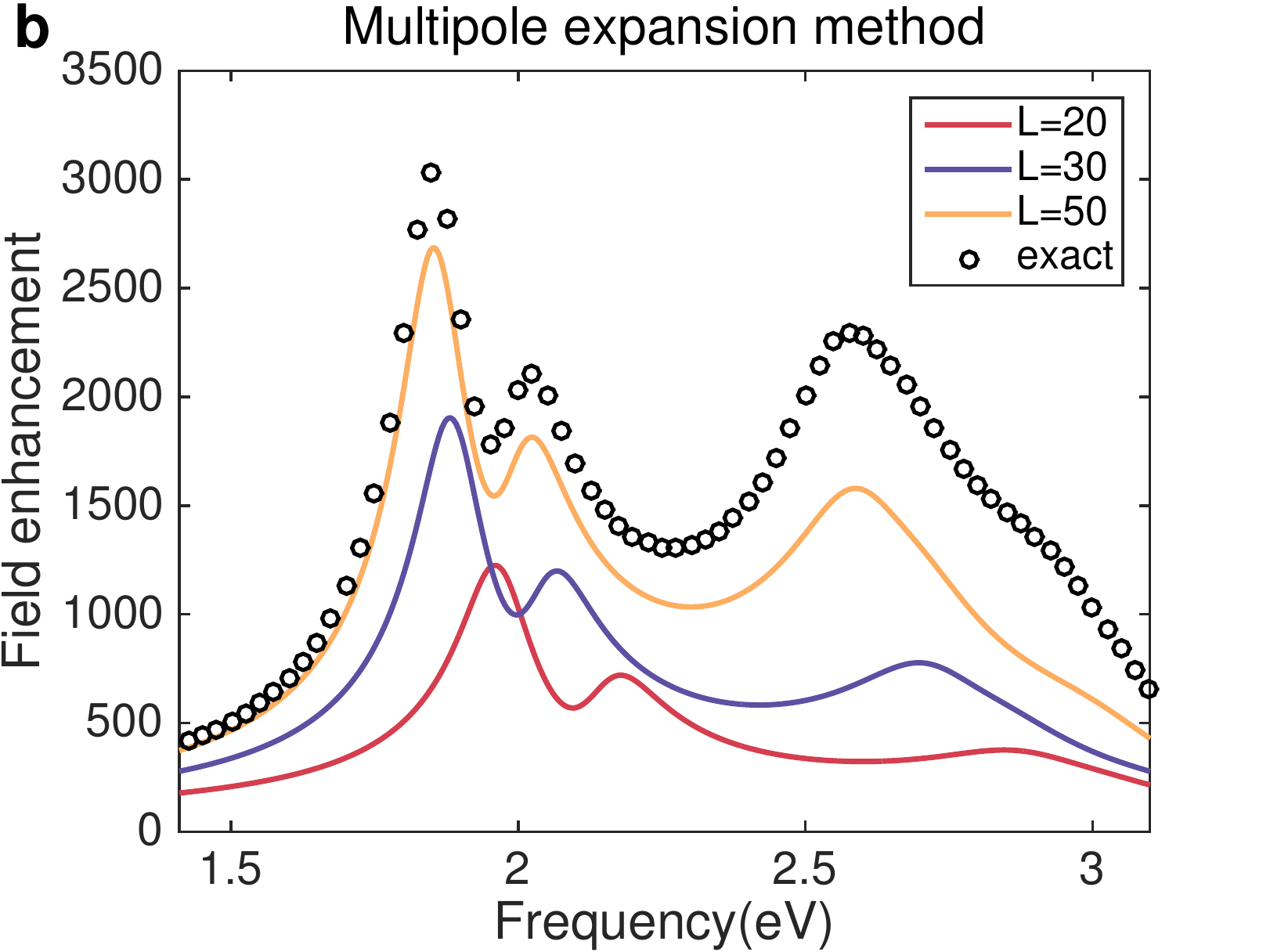}
\includegraphics[height=3.7cm]{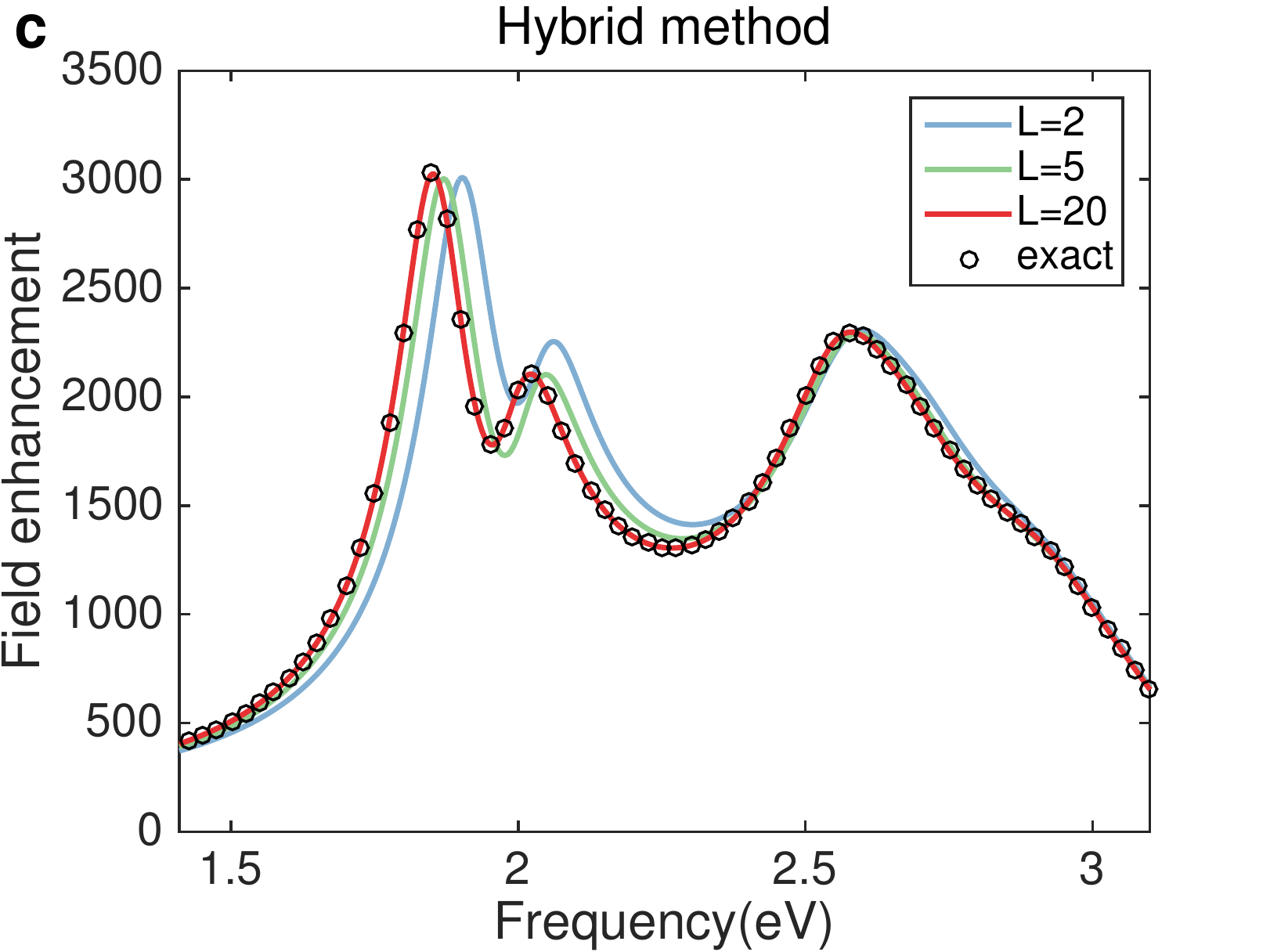}
\vskip.2cm
\includegraphics[height=3.7cm]{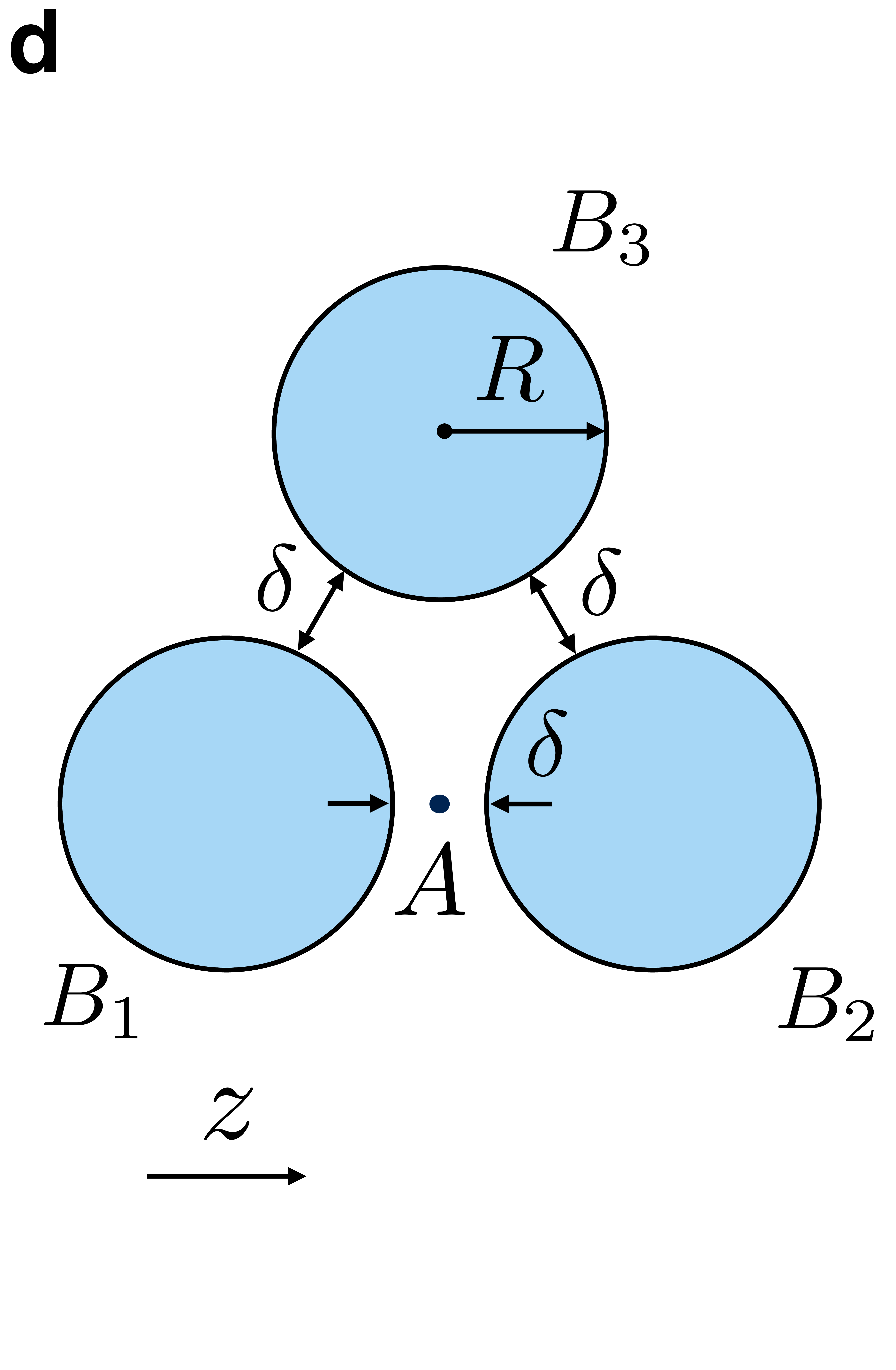}
\hskip.2cm
\includegraphics[height=3.7cm]{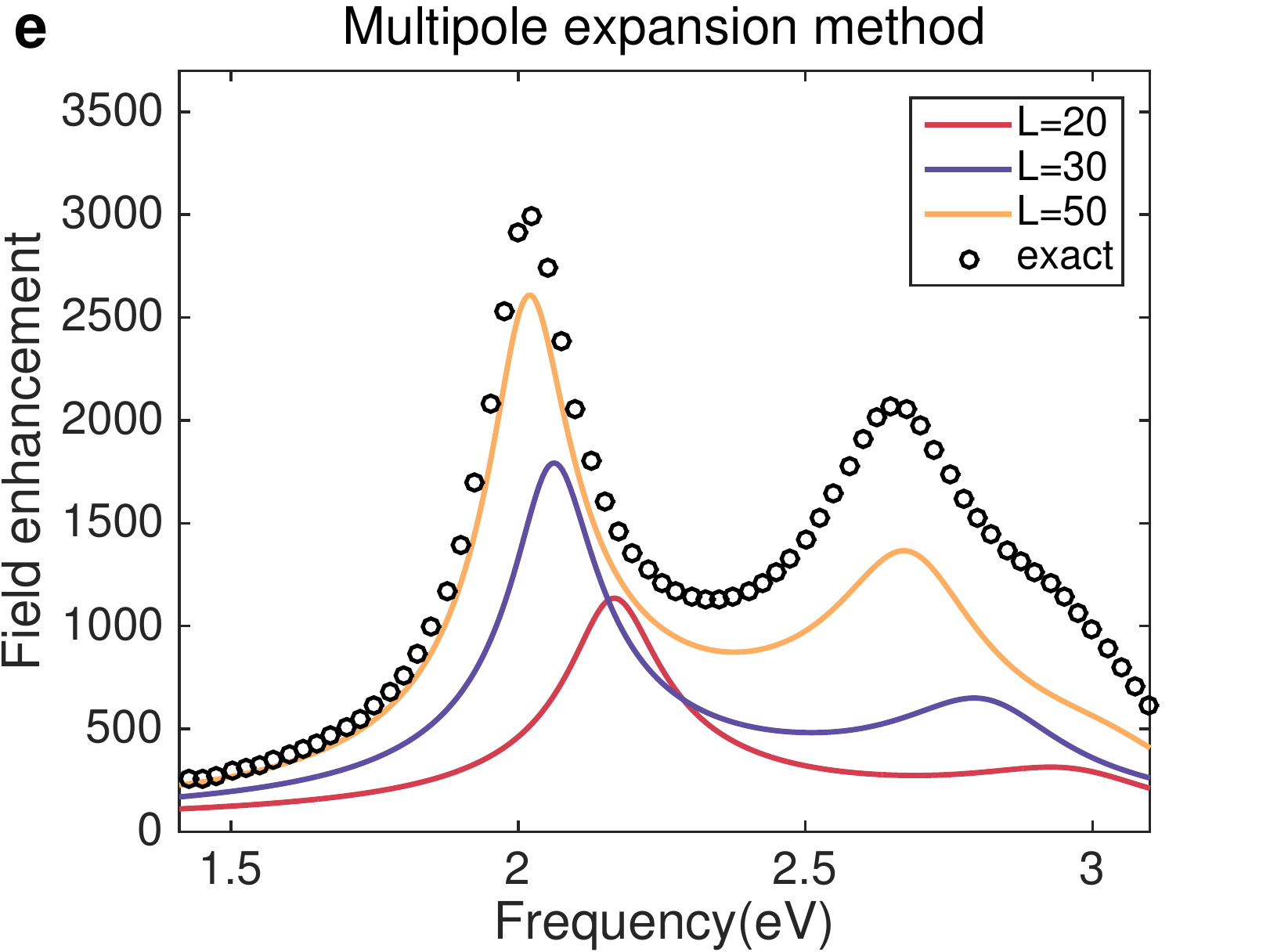}
\includegraphics[height=3.7cm]{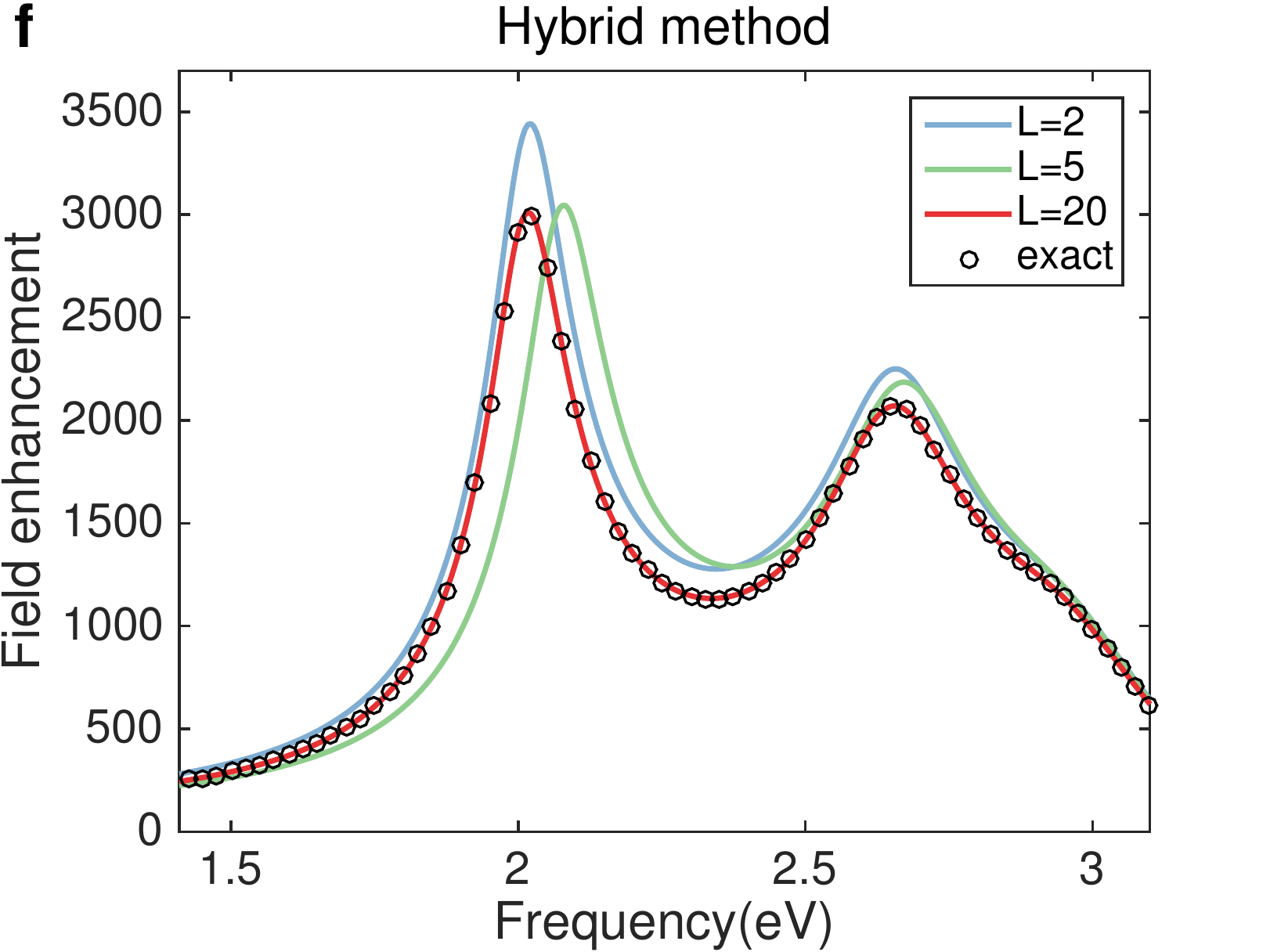}
\end{center}
\caption{ Multipole expansion method vs Hybrid scheme. (a,d) Two examples of three spheres configurations. (b,c) The field enhancement at point $A$ as a function of frequency for the configuration (a)  using the multipole expansion method and the hybrid method, respectively. The parameters
 are given as $R=30$ nm, $\delta=0.25$ nm and $\beta=80^{\circ}$. The uniform incident field $(0,0,\mbox{Re}\{ e^{i\omega t} \})$ is applied (e,f) Same as (b,c) but for the configuration (d).}
\label{fig5}
\end{figure*}

Next, we present numerical examples to illustrate the hybrid method. We consider two examples of the three-spheres configuration shown in Figures \ref{fig5}a and \ref{fig5}d. We show comparison between multipole expansion method and the hybrid method by plotting the field enhancement at the gap center $A$. For the numerical implementation, only a finite number of the multipoles $\mathcal{Y}_{lm}$ or hybrid multipoles $U_{lm}^\pm$ should be used. Let $L$ be the truncation number for the order $l$.
In Figures \ref{fig5}b and \ref{fig5}e, the field enhancement is computed using the standard multipole expansion method. The computations give inaccurate results even if we include a large number of multipole sources with $L=50$. On the contrary, the hybrid method gives pretty accurate results even for small values of $L$ such as $L=2$ and $5$; see Figures \ref{fig5}c and \ref{fig5}f. Also, $99\%$ accuracy can be achieved only with $L=20$. For each hybrid multipole $U_{lm}^\pm$, the TO harmonics are included up to order $n=300$ to ensure convergence and it can be evaluated very efficiently.

To achieve $99.9\%$ accuracy at the first resonant peak, it is required to set $L=150$ in the multipole expansion method and a $68,400 \times 68,400$ linear system needs to  be solved. However, the same accuracy can be achieved only with $L=23$ in the hybrid method. The corresponding linear system's size is $1,725 \times 1,725$ and it can be solved $2,000$ times faster than that of the multipole expansion method. In Figure \ref{fig6}, we also show the field distribution for the three-spheres examples. The high field concentration in the narrow gap regions between nanospheres is clearly shown.

\begin{figure*}
\begin{center}
\includegraphics[height=4.5cm]{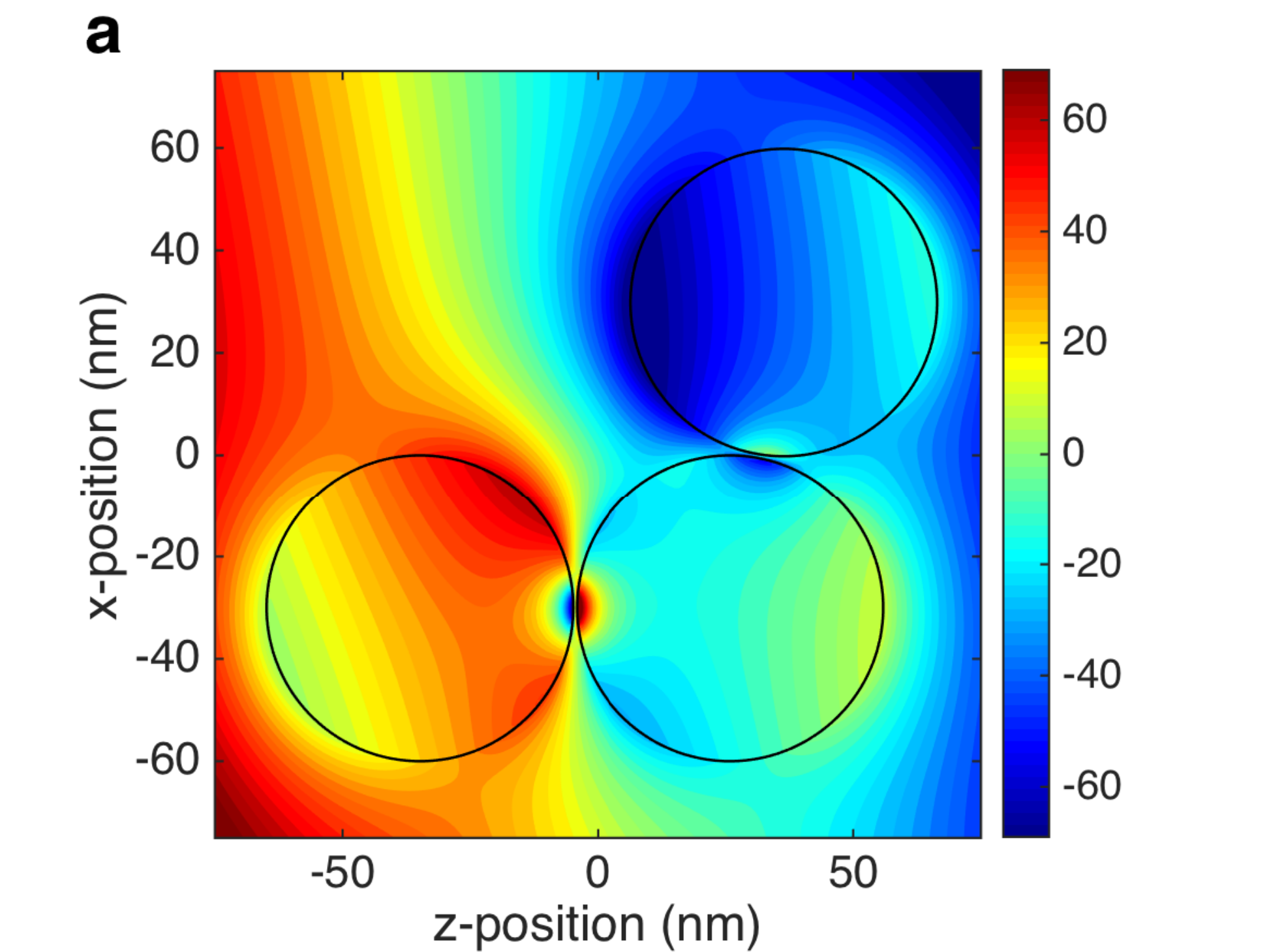}
\includegraphics[height=4.5cm]{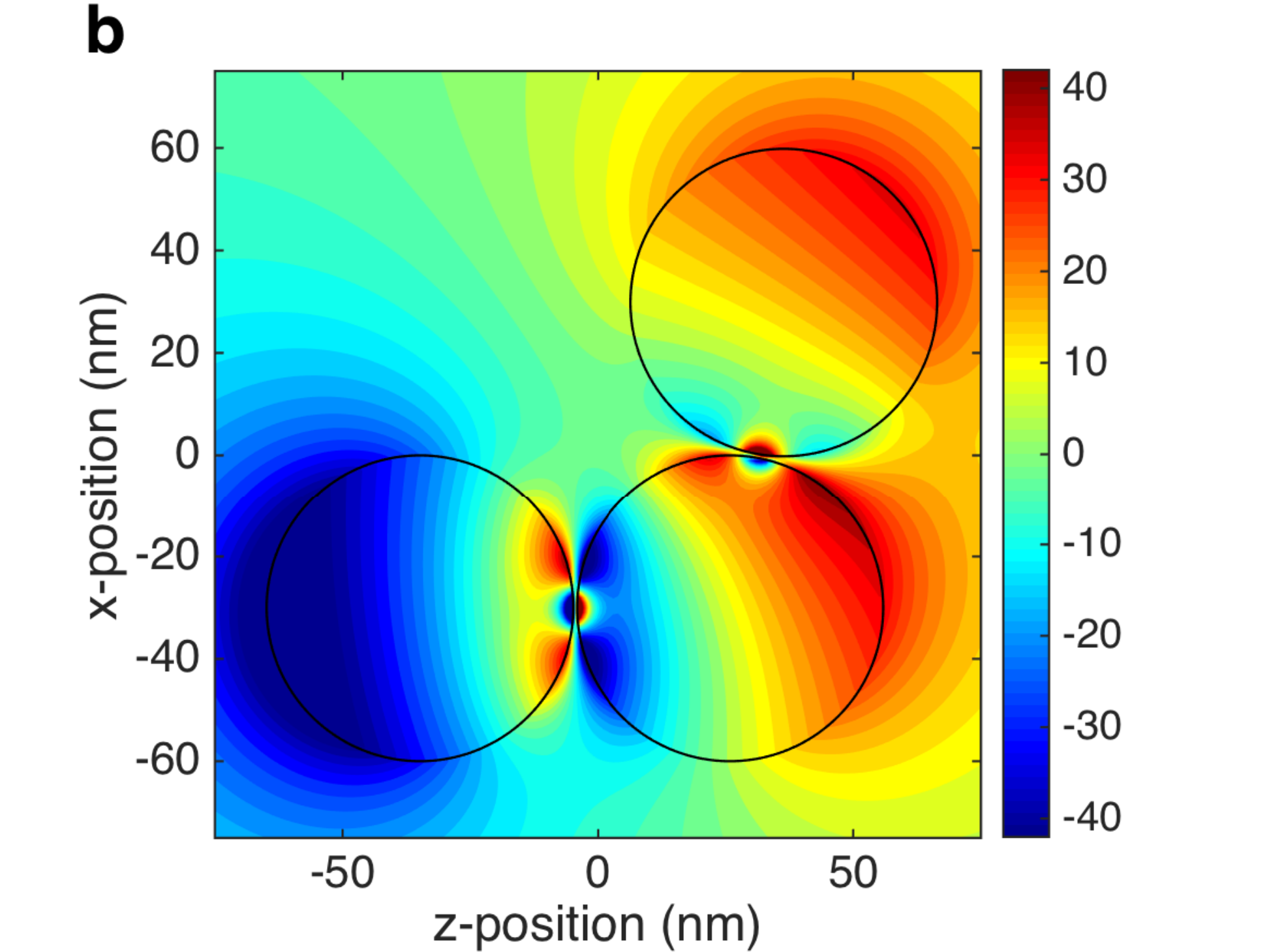}
\vskip.3cm
\includegraphics[height=4.5cm]{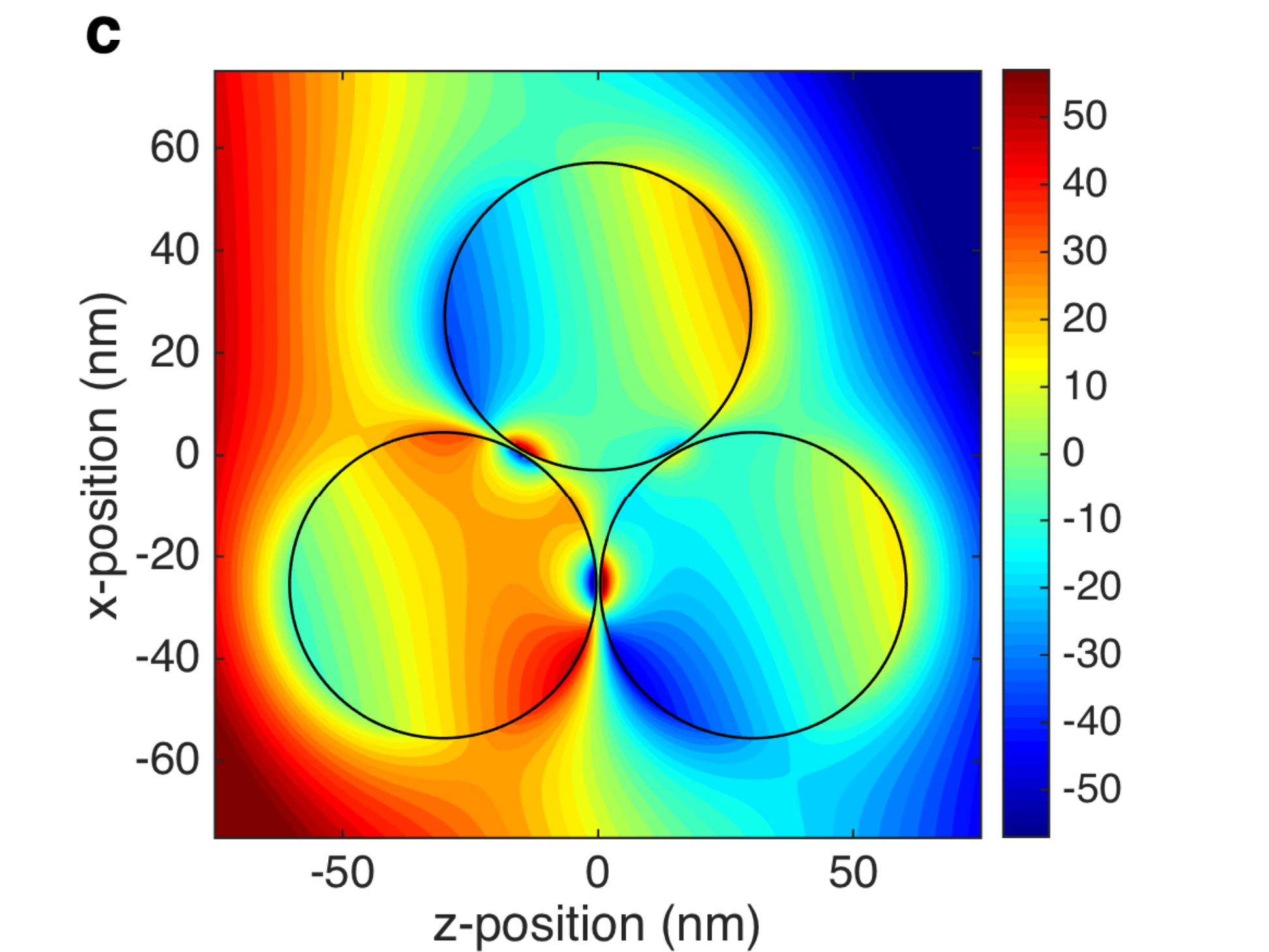}
\includegraphics[height=4.5cm]{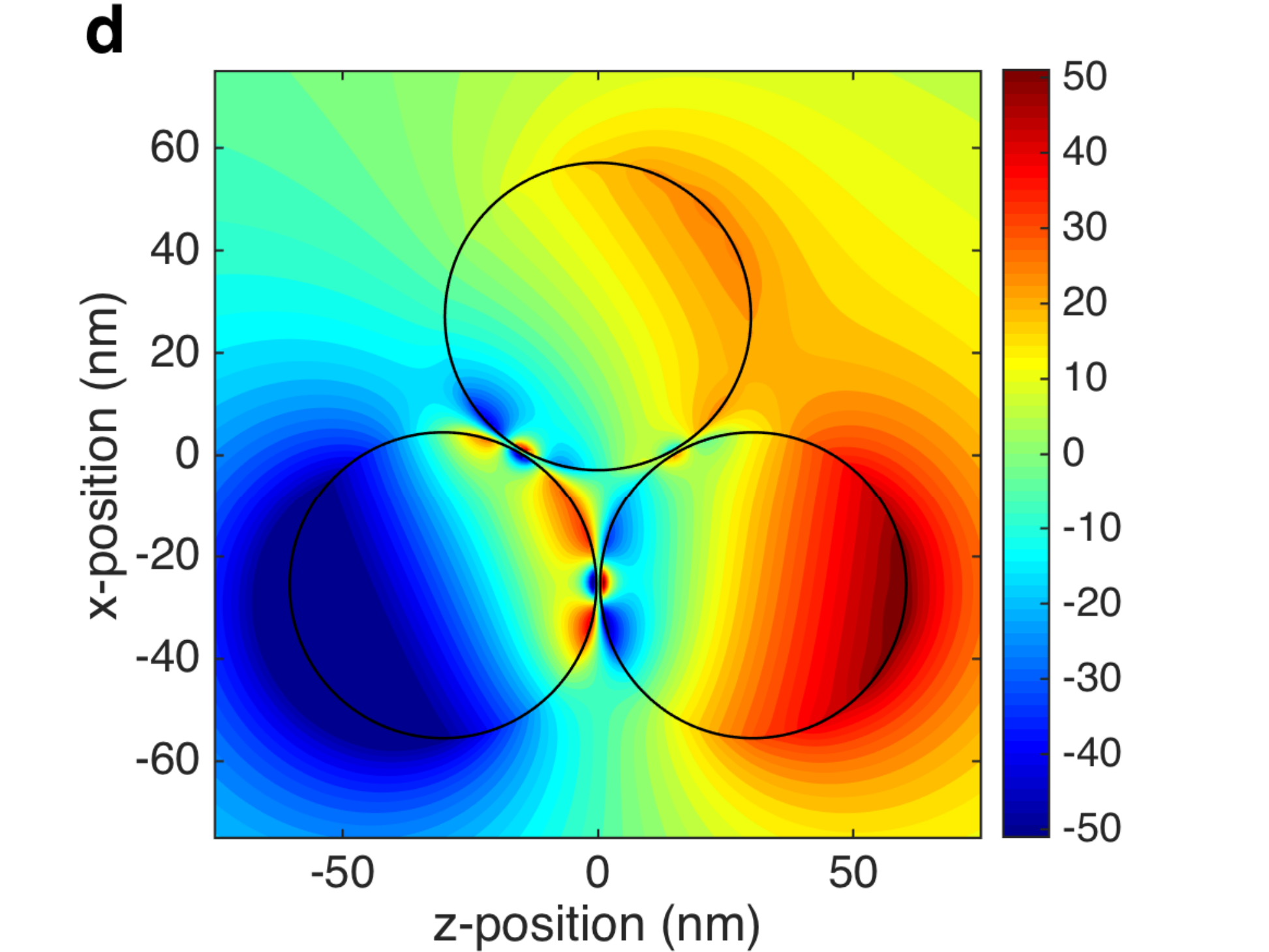}
\end{center}
\caption{ Potential distributions for the two three spheres examples. 
(a,b) Real and imaginary parts of the potential for the configuration in Figure \ref{fig5}a with $R=30$ nm, $\delta=0.25$ nm, and $\beta=80^\circ$. 
The uniform incident field $(\sin 15^{\circ},0,\cos 15^{\circ})\mbox{Re}\{ e^{i\omega t}\}$
 is applied at $\omega=3.0$ eV. 
(c,d) Same as (a,b) but for the configuration in Figure \ref{fig5}d.}
\label{fig6}
\end{figure*}

\section{Discussion} 
In this paper, we have fully characterized the singular nature of interaction between  nearly touching plasmonic nanospheres in an analytical way. Based on new analytic solutions, we also have extended Cheng and Greengard's hybrid numerical method to the case of plasmonic spheres. The extended scheme gives extreme efficiency and accuracy for computing the field generated by an arbitrary number of plasmonic spheres even when  they are  nearly touching. We have assumed that the spheres are identical only for simplicity. Our approach can be  directly extended to the case where the spheres are not equisized and have different material parameters. Moreover, by coupling with the fast multipole method, we expect that the proposed scheme will give an efficient numerical solver for a large scale problem \cite{fmm,fmm_hybrid}. In that case, we should consider the retardation effect which comes from the finite wavelength of the incident light. But the quasistatic interaction is dominant for any pair of closely spaced spheres and so our result is still useful in a large scale problem. This will be the subject of a forthcoming paper. The nonlocal effect, which has a quantum mechanical origin, is an important issue when the spheres are extremely closely spaced \cite{Luo_nonlocal,Nordlander_nonlocal,Ory2}. By adopting the shifting boundary method developed by Luo et al. \cite{Luo_nonlocal}, this effect can be easily incorporated.

\appendix

\section{Some definitions and properties} \label{sect2}

\begin{itemize}
\item Let us define the spherical harmonics $Y_{l}^m$ by
\begin{equation*}
Y_{l}^m(\theta,\phi) = \sqrt{\frac{(l-|m|)!}{(l+|m|)!}} P_l^{|m|}(\cos\theta)e^{im\phi},
\end{equation*}
where $P_l^m(x)$ is the associated Legendre polynomial given by
\begin{equation*}
P_l^m(x) = (-1)^m (1-x^2)^{m/2} \frac{d^m}{dx^m}P_l(x).
\end{equation*}
Here, $P_l(x)$ is the Legendre polynomial of degree $l$.

\item The Legendre polynomial $P_n(x)$ has the following generating function:
\beq\label{gen_P0}
\frac{1}{\sqrt{1-2ax+a^2}}=\sum_{n=0}^\infty a^n P_n(x).
\eeq
\item The associated Legendre polynomial $P_n^m(x)$ has the following generating function:
\beq\label{gen_Pnm}
(-1)^m(2m-1)!!\frac{(1-x^2)^{m/2}a^m}{[1-2ax+a^2]^{m+1/2}} = \sum_{n=m}^\infty a^n P_n^m(x).
\eeq
\item It holds that
\beq\label{Pnn}
P_{n}^n(x) =(-1)^n (2n-1)!!(1-x^2)^{n/2}.
\eeq

\item 
Let us define the solid harmonics $\mathcal{Y}_{lm}$ and $\mathcal{Z}_{lm}$ by
\begin{align*}
\mathcal{Y}_{lm}(\Br) &= r^{-(l+1)}Y_{l}^m(\theta,\phi),
\\
\mathcal{Z}_{lm}(\Br) &= r^l\,Y_{l}^m(\theta,\phi).
\end{align*}
The function $\mathcal{Y}_{lm}$ is also called the multipole source.
   
\item Let us introduce
\beq\label{def_wm_app}
w_{m}=
\begin{cases}
\ds 1, &\quad m \geq 0,
 \\
\ds (-1)^{|m|} , &\quad m < 0.
\end{cases}
\eeq

\item Let the constant $N_{lmab}$ be given by
\beq\label{def_Nlmab_app}
N_{lmab} = (-1)^{a+b}
\sqrt{
\begin{pmatrix} l+a-b+m \\ l+m\end{pmatrix}
\begin{pmatrix} l+a+b-m \\ a+b\end{pmatrix}
},
\eeq
and let the constant $N_{lm}$ be given by
\beq\label{def_Nlm_app}
\ds N_{lm}=(l-|m|)!
\sqrt{
\begin{pmatrix} l+|m| \\ l+m\end{pmatrix}
\begin{pmatrix} l+|m| \\ |m|+m\end{pmatrix}
}.
\eeq

\end{itemize}

\section{TO inversion mapping and the bispherical coordinates} \label{sect1}
Here we present the basic properties of the TO inversion mapping.
The TO inversion mapping $\Phi$ can be rewritten using the bispherical coordinates $(\xi,\eta,\varphi)$ defined by
\begin{equation} \label{def_bispherical_app}
e^{\xi-i\eta} = ({     {  z+i\rho }    +\alpha})/({      {  z+i\rho   } -\alpha }),
\end{equation}
with  $\rho=\sqrt{x^2+y^2}$ and $\varphi$ being the azimuthal angle.

By letting $\mathbf{r}'=e^{\xi}(\sin\eta\cos\varphi,\sin\eta\sin\varphi,\cos\eta)$, $\mathbf{R}_0'=(0,0,1)$, $\mathbf{R}_0 = (0,0,\alpha)$ and $R_T^2=2\alpha$, we can see that the bispherical transformation  is identical to the inversion mapping  in the TO approach. Although they are the same, it is worth mentioning that TO approach gives novel physical insights into the interaction between two plasmonic spheres. The reason why we rewrite $\Phi$ in terms of the bipsherical coordinates is that many useful properties have been derived in this coordinate system. In Figure \ref{fig1_app}, the geometry of the bispherical coordinates is described.

\begin{figure*}
\begin{center}
\includegraphics[height=6.0cm]{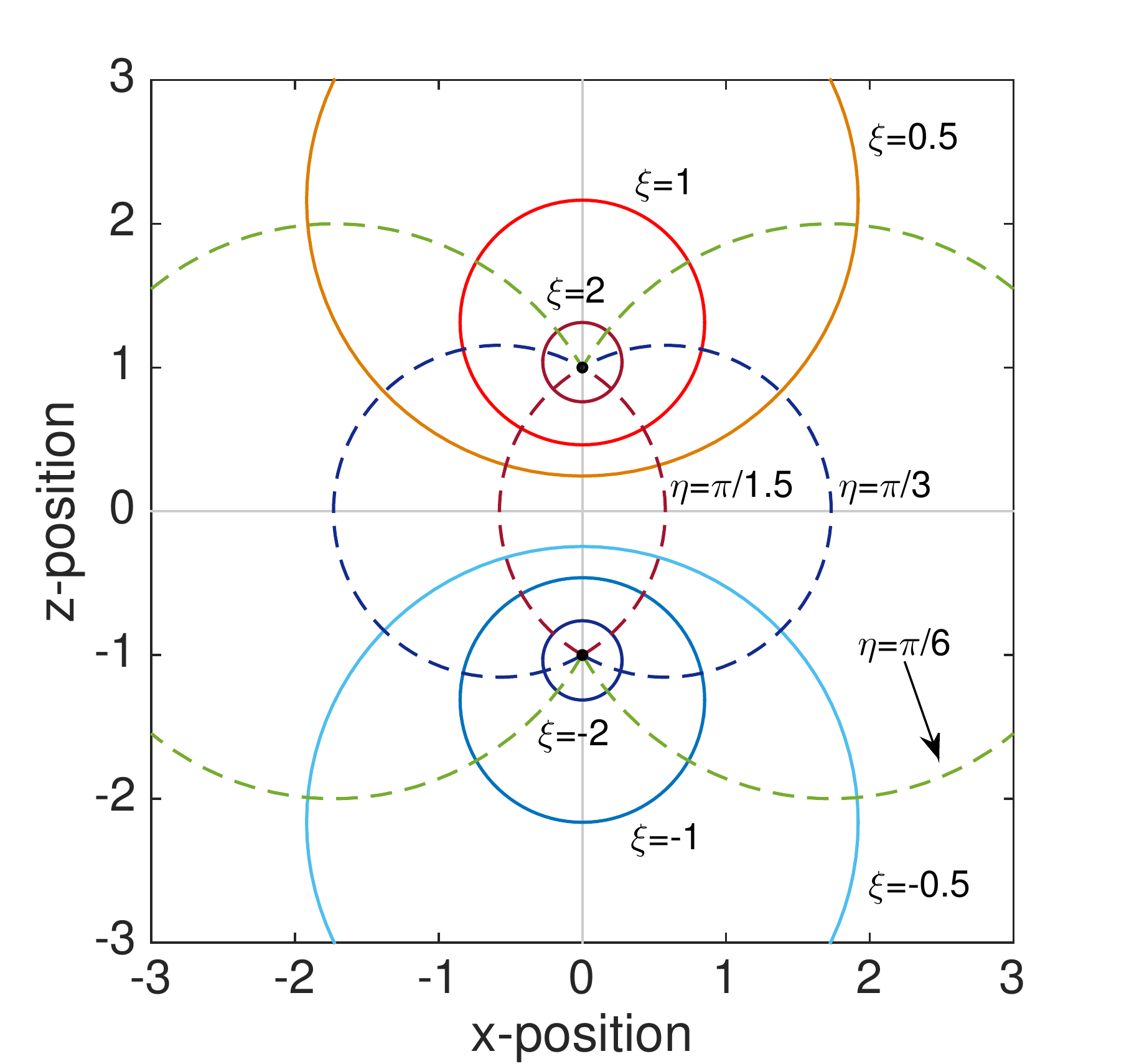}
\end{center}
\caption{Bispherical coordinates.
Coordinate level curves for the bispherical coordinate system with $\alpha=1$. The solid lines (resp. the dashed lines) represent $\xi$ (resp. $\eta$) coordinate curves.}
\label{fig1_app}
\end{figure*}

The Cartesian coordinates can be written in terms of the bispherical ones as follows: \beq\label{bisph_cartes_app}
\ds x=\frac{\alpha\sin\eta\cos\varphi}{\cosh\xi-\cos\eta},\
\ds y=\frac{\alpha\sin\eta\sin\varphi}{\cosh\xi-\cos\eta},\
\ds z=\frac{\alpha\sinh\xi}{\cosh\xi-\cos\eta}.
\eeq
Note that the origin $(0,0,0)$ corresponds to $\xi=0,\eta=\pi,\varphi=0$. The point at infinity corresponds to $(\xi,\eta)\rightarrow(0,0)$. On the other hand, it can be easily shown that the coordinate surfaces $\{\xi=c\}$ and $\{\eta=c\}$ for a nonzero $c$ are respectively the zero level set of
\begin{align}
 f^\xi(x,y,z)&=
(z-\alpha \coth c)^2 +\rho^2-\left({\alpha}/{\sinh c}\right)^2, \label{fxi}
\\
 f^\eta(x,y,z)&=
(\rho-{\alpha}{\cot c})^2+z^2 -\left({\alpha}/{\sin c}
\right)^2.
\end{align}
Note also that the $\xi$-coordinate surface is the sphere of radius $\alpha/\sinh c$ centered at $(0,0,\alpha \coth c)$. Therefore, $\xi=c$ (or $\xi=-c$) represents a sphere contained in the region $z>0$ (resp. $z<0$). Moreover, $|\xi|<c$ (resp. $|\xi|>c$) represents the region outside (resp. inside) the two spheres. 

Suppose that two spheres $B_+$ and $B_-$ of the same radius $R$ are centered at $(0,0,+d)$ and $(0,0,-d)$, respectively. We set $s=\cosh^{-1}(d/R)$ and $\alpha=R\sinh s$. Then we have $d=\alpha \coth s$ and $R=\alpha/\sinh s$. So, in view of \cref{fxi}, the surfaces $\{\partial B_\pm\}$ of the two spheres are parametrized by $\{\xi=\pm s\}$, respectively. 

Any solution to the Laplace's equation can be represented as a sum of the following bispherical harmonics which is equal to the TO basis $\mathcal{M}_{n,\pm}^m$ as follows: 
\beq\label{TO_basis_bispherical}
\mathcal{M}_{n,\pm}^m(\Br) = \sqrt{2}\sqrt{\cosh\xi-\cos\eta}\,e^{\pm(n+\frac{1}{2})\xi}Y_{n}^m(\eta,\varphi).
\eeq
The scale factors for the bispherical coordinates are
\beq
\sigma_\xi=\sigma_\eta=\frac{\alpha}{\cosh \xi-\cos\eta}\quad\mbox{and}\quad \sigma_\varphi = \frac{\alpha\sin\eta}{\cosh \xi-\cos\eta},
\nonumber
\eeq
so that the gradient for a scalar valued function $g$ can be written in the form
\beq
\nabla g = \frac{1}{\sigma_\xi}\frac{\partial g}{\partial\xi}{\mathbf{e}}_{\xi}+ \frac{1}{\sigma_\eta} \frac{\partial g}{\partial\eta}{\mathbf{e}}_{\eta}+\frac{1}{\sigma_\varphi }\frac{\partial g}{\partial\varphi}{\mathbf{e}}_{\varphi},
\nonumber
\eeq
where $\{\mathbf{e}_\xi,\mathbf{e}_\eta,\mathbf{e}_\varphi\}$ are the unit basis vectors in the bispherical coordinates. The normal derivative on the surface $\{\xi=\pm s\}$ of the sphere $B_\pm$ is given by
\beq\label{normal_deri}
\frac{\partial}{\partial\mathbf{n}}\Big|_{\partial B_\pm} = \mp{\mathbf{e}}_\xi\cdot \nabla|_{\partial B_\pm}= \mp \frac{\cosh s-\cos\eta}{\alpha}\frac{\partial}{\partial \xi}\Big|_{\xi=\pm s},
\eeq
where $\mathbf{n}$ denotes the outward unit normal vector.

If the function $g$ is of the following form:
\beq
\nonumber
g(\Br)=\sum_{n=0}^\infty \sum_{m=-n}^n c_n^m \mathcal{M}_{n,+}^m(\Br)+d_n^m \mathcal{M}_{n,-}^m(\Br),
\eeq
then $z$-component of the gradient at the origin is given by
\beq\label{grad_origin_bisph}
{\mathbf{e}}_z\cdot \nabla g (0,0,0) = \frac{2^{3/2}}{\alpha}\sum_{n=0}^\infty (c_n^0-d_n^0) (n+1/2) (-1)^n,
\eeq
where ${\mathbf{e}}_z=(0,0,1)$.

\section{Recurrence relations for $A_n$}\label{app_sec:recur_An}
In \cite{GN}, it was shown that the coefficients $A_n$ in \cref{TO_V_eqn_main} satisfy the following recurrence relations: for $n=0,1,2,\cdots$,
\begin{align*}
(\overline{n}-1/2) f_{n-1} A_{n-1}
- g_n A_n + (\overline{n}+1/2)f_{n+1} A_{n+1}
=E_0 h_n,
\end{align*}
where $\overline{n}=n+1/2$, $A_{-1}=0$ and $f_n,g_n,h_n$ are given by
\begin{align*}
f_n &= \cosh (\overline{n}-1)s + \ep_B \sinh (\overline{n}-1)s,
\\
g_n &= 2\overline{n}\cosh s(\sinh \overline{n} s + \cosh \overline{n} s) +(\ep_B-1)\sinh s \sinh \overline{n}s , 
\\
h_n &= \alpha(1-\ep_B) e^{-\overline{n}s}(ne^s - (n+1) e^{-s} ).
\end{align*}
For the reader's convenience, we recall $s=\cosh^{-1}(d/R)$ and $\alpha=R\sinh s$.

\section{Poladian's image method for two spheres (review)}\label{sec:Poladian}

Here, we present a review of Poladian's image method for two dielectric spheres. First, we explain the image method when only a single sphere is placed in the whole space. Then we discuss an image series solution for two spheres in a uniform incident field. Finally, we consider the generalized image method for the case of multipole sources.

\subsection{A single sphere}
Suppose that a single sphere of radius $R$ is centered at $(0,0,0)$. Let $\ep_B$ be the permittivity of the sphere. We also assume the background permittivity is $\ep_0=1$. Let $\tau=(\ep_B-1)/(\ep_B+1)$. When we locate a point charge $Q$ at $(0,0, c)$ with $c>R$, then it can be shown that the reaction potential is identical to the potential generated by the following two image sources \cite{Pol_thesis,Pol,Pol2}: (1) a point charge $Q'=-\tau (R/c) Q$ at $(0,0,R^2/c)$ and (2) a continuous line source along the line segment from $(0,0,0)$ to $(0,0,R^2/c)$ with a density function $\Lambda$ given by
\begin{equation*}
\Lambda(t) = \frac{\tau Q}{R(\ep+1)}  \Big(\frac{R^2}{c t}\Big)^{\frac{1}{2}(\tau+1)}, \quad t\in (0,R^2/c) .
\end{equation*}
Poladian observed that the continuous line source can be well approximated by a point charge $-Q'$ at the center of the sphere $(0,0,0)$ provided that $|\ep_B|$ is large. In fact, this approximation becomes exact when $|\ep_B|=\infty$.

Therefore, Poladian's imaging rule for a single sphere can be summarized as follows: if a sphere of radius $R$ is centered at $(0,0,0)$ and a point charge $Q$ is located at $(0,0,c)$, then the following two image charges are produced: (i) a point charge $Q'=-\tau (R/c)Q$ at $(0,0,R^2/c)$ (ii) a point charge $-Q'$ at the center of the sphere $(0,0,0)$ \cite{Pol_thesis,Pol,Pol2}. Let us call the latter image charge {\it the neutralizing charge}.

\subsection{Two spheres in a uniform field}\label{app_Poladian_uniform}

Let us now consider the two spheres $B_+$ and $B_-$. Suppose that we locate a point charge of the magnitude $\pm1$ at the point $(0,0,\pm z_0)$ in the sphere $B_\pm$, respectively. Due to the interaction between two spheres, an infinite sequence of image charges is generated along $z$-axis by Poladian's imaging rule. But it is difficult to keep track of all the image charges at each step of the recursive imaging process. Poladian found that it is much simpler to initially neglect the neutralizing charges and later introduce an additional image sources.

By ignoring the neutralizing charge in Poladian's imaging rule, we obtain an infinite sequence of the image charges as follows: for $m=0,1,2,\cdots$, $m$-th image charge $\pm u_m$ is located at the point $\pm\mathbf{z}_m=(0,0,\pm z_m)$ in the sphere $B_\pm$, respectively, where $z_m$ and $u_m$ satisfy the following recursive relations:
\beq
\nonumber
\ds d-z_{k+1} = \frac{R^2}{d+z_k},
\quad
\ds u_{k+1} = \tau\frac{R}{d+z_k} u_k . 
\eeq
These recursive relations can be solved explicitly. To state the solutions for $u_k$ and $z_k$, we introduce a parameter $t_0$ which satisfies $z_0 = \alpha \coth (s+t_0)$. Note that if the initial position is equal to the center of each sphere (that is, $z_0=d=R\cosh s$), then it holds that $t_0=0$. Using this representation for $z_0$ and the hyper-trigonometric identities, one can see that the solutions for $z_k$ and $u_k$ are given as follows:
\begin{equation*}
\ds{z_k} = \alpha \coth (ks+s+t_0),
\quad
\ds u_k = \tau^{k}  \frac{\sinh (s+t_0)}{\sinh (ks+s+t_0)}.
\end{equation*}
Then the potential $U(\Br)$ generated by all the above image charges is given by
\beq\label{image_series_U_app}
U(\Br) = \sum_{k=0}^\infty u_k ( G(\Br-\mathbf{z}_k) - G(\Br+\mathbf{z}_k) ),
\eeq
where $\mathbf{z}_k=(0,0,z_k)$.

Let  us now consider the two spheres $B_+\cup B_-$ placed in a uniform incident field  $(0,0,E_0)$. Let $p_0$ be the induced polarizability when a single sphere is subjected to  the uniform incident field, that is, $p_0 =E_0 R^3{2\tau }/({3-\tau}).$ We also let $D(\Br) ={\mathbf{e}_z \cdot \hat{\mathbf{r}}}/({|\Br|^2})$ be the potential generated by a point dipole source with a unit moment $\mathbf{e}_z$, where $\hat{\Br} = \Br/|\Br|$. The uniform incident field is first imaged in each sphere, inducing an image point dipole source with the polarizability $p_0$ at the center of each sphere. Then these initial point dipoles produce an infinite sequence of image sources. The point dipole $p_0$ can be considered as the limit of two initial charges $\pm {4\pi p_0}/{2h}$ at the points $z_0=(0,0,d\pm h)$ as $h\rightarrow 0$. It is equivalent to taking derivative $4\pi p_0\partial/{\partial z_0}$ at $z_0=d$. So we get the following expression for the image potential generated by the point dipole $p_0$ \cite{Pol_thesis,Pol,Pol2}:
\beq\label{V1_def}
V_1(\Br):= 4\pi p_0 \frac{\partial  (U(\Br))}{\partial z_0}\Big|_{z_0=d}.
\eeq
Then, from \cref{image_series_U_app} and the identity $\frac{\partial}{\partial{z_0}}|_{z_0=d} = -\frac{\sinh^2 s}{\alpha} \frac{\partial}{\partial{t_0}}|_{t_0=0}$, we can represent $V_1$ more explicitly as follows:
\begin{equation*}
\ds V_1(\Br)= \sum_{k=0}^\infty p_k D(\Br-\Br_k)- q_k G(\Br-\Br_k)
+\sum_{k=0}^\infty p_k D(\Br+\Br_k)+ q_k G(\Br+\Br_k),
\end{equation*}
where $\Br_k,p_k$ and $q_k$ are given by
\begin{align*}
\ds\Br_k &= \mathbf{z}_k|_{t_0=0}=(0,0, \alpha \coth(k+1)s),
\\[0.2em]
\ds p_k &= \tau^{k} p_0 \Big( \frac{\sinh s}{\sinh (k+1)s} \Big)^3,
\quad q_k =  \tau^{k}\frac{p_0}{R}\frac{\sinh s \sinh k s}{\sinh^2 (k+1)s}.
\end{align*}
Note that $\pm \mathbf{r}_0$ is the center of the sphere $B_\pm$, respectively.

As pointed out by Poladian in \cite{Pol_thesis}, the potential $V_1$ is unphysical because the total charge on each sphere is non-zero. It originates from the fact that we have ignored the neutralizing image charges. Now we explain Poladian's strategy for neutralizing the total charge \cite{Pol_thesis}. We introduce an additional potential by locating a point charge $\pm Q$ at the center of the sphere $B_\pm$, respectively. Then the corresponding image potential is given by
\begin{equation*}
\ds V_2(\Br):= Q U(\Br)|_{z_0=d}
=  Q \sum_{k=0}^\infty u_k^0 ( G(\Br-\Br_k)-G(\Br+\Br_k) ),
\end{equation*}
where $u_k^0$ is defined by $u_k^0 = u_k|_{t_0=0} = \tau^k {\sinh s}/{\sinh (k+1)s}$. Now we choose the constant $Q$ so that the potential $V_1+V_2$ has no net flux on each sphere. More precisely, we impose $\int_{\partial B_\pm} \frac{\partial }{\partial\mathbf{n}}(V_1+V_2) dS = 0$. Then we obtain $Q =\sum_{k=0}^\infty q_k \Big/\sum_{k=0}^\infty u_k^0$. Finally, we get the approximation for the potential $V(\Br)$ by superposing the uniform incident field and the aformentioned potentials:
\beq\label{Poladian_sol_V}
V(\Br)+E_0 z \approx V_1(\Br)+V_2(\Br)=4\pi p_0 \frac{\partial  (U(\Br))}{\partial z_0}\Big|_{z_0=d} + Q U(\Br)|_{z_0=d}.
\eeq

\subsection{Image method for general multipole sources} \label{app_Poladian_multipole}

We now consider the case when an initial image source is a multipole source $\mathcal{Y}_{lm}$. Note that, since the point charge potential $G$ and the dipole potential $D$ satisfy $G(\Br)=\frac{1}{4\pi}\mathcal{Y}_{00}(\Br)$ and $D(\Br)=\mathcal{Y}_{10}(\Br)$, the image potentials \cref{image_series_U} and \cref{V1_def} can be seen as the special cases of potentials generated by the image multipole sources.

Before considering a general multipole source $\mathcal{Y}_{lm}$, let us first consider a sectoral multipole $\mathcal{Y}_{|m|,m}$. If an initial sectoral multipole $\mathcal{Y}_{|m|,m}$ is located at $(0,0, z_0)$, the image sequence is produced by Poladian's rule \cite{Pol_thesis, Pol,Pol2} as follows: $u_m^{(2k)} \mathcal{Y}_{|m|,m}$ at $(0,0,z_{2k})$ and $-u_m^{(2k+1)}\mathcal{Y}_{|m|,m}$ at $(0,0,- z_{2k+1})$ for $k=0,1,2,\cdots$. Similarly, if an initial location is  $(0,0, -z_0)$, then the following image sequence is produced: $u_m^{(2k)} \mathcal{Y}_{|m|,m}$ at $(0,0,-z_{2k})$ and $-u_m^{(2k+1)}\mathcal{Y}_{|m|,m}$ at $(0,0,+ z_{2k+1})$ for $k=0,1,2,\cdots$. Here, $u_m^{(k)}$ satisfies a recursive relation 
\begin{equation*}
u_m^{(k+1)} = \tau\Big(\frac{R}{d+z_k}\Big)^{2|m|+1} u_m^{(k)}, \quad k=0,1,2,\cdots.
\end{equation*}
It can be explicitly solved as follows:
\begin{equation*}
u_{m}^{(k)}=\tau^k \Big(\frac{\sinh(s+t_0)}{\sinh(ks+s+t_0)}\Big)^{2|m|+1}.
\end{equation*}
Let $U_{m}^\pm$ be the potential generated by the above image sequence when the initial sectoral multipole is located at $(0,0,\pm z_0)$, respectively.
Then the potential $U_{m}^\pm$ is given by
\beq\begin{array}{l}
\ds U_{m}^{\pm}(\Br)=\sum_{k=0}^\infty u_{m}^{(2k)} \mathcal{Y}_{|m|,m}(\Br \mp \mathbf{z}_{2k})
-u_{m}^{(2k+1)} \mathcal{Y}_{|m|,m}(\Br\pm\mathbf{z}_{2k+1}).\label{Umm_image_sol}
\end{array}\eeq

Let us turn to the case of a general multipole source $\mathcal{Y}_{lm}(\Br)$. Let $U_{lm}^\pm$ be the potential due to the image sequence produced by an initial multipole source $\mathcal{Y}_{lm}$ located at the center of the sphere $B_\pm$, respectively. It was shown that a general multipole source $\mathcal{Y}_{lm}$ can be represented as a derivative of a sectoral multipole $\mathcal{Y}_{|m|,m}$\cite{Pol_thesis,Pol,Pol2}:
\beq\label{Ylm_and_Ymm}
\mathcal{Y}_{lm}(\Br\mp\Br_0) = \mathcal{D}_{lm}^\pm \big[ \mathcal{Y}_{|m|,m}(\Br\mp\mathbf{z}_0)\big],
\eeq
where the differential operator $\mathcal{D}_{lm}^\pm$ is defined by
\beq\label{eqn:D_lm_pm}
\ds \mathcal{D}_{lm}^\pm[f]=\frac{(\pm 1)^{l-|m|}}{N_{lm}} \frac{\partial^{l-|m|}}{\partial  [z_0(t_0)]^{l-|m|}}f\bigg|_{z_0=d},
\eeq
where $N_{lm}$ is defined as \cref{def_Nlm_app}. Therefore, the image potential $U_{lm}^\pm$ is also represented as $U^{\pm}_{lm}(\Br) = \mathcal{D}_{lm}^\pm\big[ U^{\pm}_{m}(\Br)\big]$. Actually, this is not the end. We need to be careful when we consider the case when $m=0$. In this case, the total charges on each sphere $B_\pm$ may be non-zero. Since this is unphysical, we have to neutralize them again. We modify the potential $U_{lm}^+$ by adding an image potential produced by the following initial charges: a point charge $-Q^{+}_{l,1}$ (and $-Q^{+}_{l,2}$) at the center of the sphere $B_+$ (and $B_-$), respectively. We also modify the potential $U_{lm}^-$ in a similar way with the initial charges $-Q^{-}_{l,i},i=1,2$.  Here, the constants $Q_{l,i}^\pm$  are chosen so that the total flux on each surface $\partial B_\pm$ is zero. Specifically, the potential $U^{\pm}_{l,m}$ is modified as follows:
\beq\label{ULM_image_sol_app}
U^{\pm}_{lm}(\Br) = \mathcal{D}_{lm}^\pm\big[ U^{\pm}_{m}(\Br)\big] -\delta_{0m} Q_{l,1}^\pm U_{0}^+(\Br)|_{z_0=d} -\delta_{0m} Q_{l,2}^\pm U_{0}^-(\Br)|_{z_0=d},
\eeq
where $\delta_{lm}$ is the Kronecker delta.

It is worth to remark that we can evaluate the derivatives $\mathcal{D}_{lm}^\pm[f]$ analytically by using the Fa\'{a} di Bruno's formula (we omit the details). Moreover, its numerical computation can be done efficiently using a recursive relation for Bell polynomials.

\section{Proof of Lemma \ref{lem_connect_charge_TO}}\label{app_single_image_TO}
From the definition \cref{def_bispherical_app} of the bispherical coordinates, we have $z+i\rho={2\alpha}/({e^{\xi-i\eta}-1})+\alpha$. Note that $\coth t={2}/({e^{2t}-1})+1$. Hence, by using these identities and letting $\mathbf{z}(t)=(0,0,\alpha \coth t)$, it follows that
\begin{align}
\ds \frac{1}{|\Br-\mathbf{z}(t)|}
&=\big|z+i\rho-\alpha\coth t\big|^{-1}
=\frac{1}{2\alpha}\left|\frac{1}{e^{\xi-i\eta}-1}-\frac{1}{e^{2t}-1}\right|^{-1}
\label{inv_distance_bisph}
\\
\ds
&=\frac{1}{2\alpha}\left|\frac{(e^{2t}-1)(e^{\xi-i\eta}-1)}{e^{2 t}(e^{\xi-2t-i\eta}-1)}\right|
=\frac{\sinh|t|}{\alpha}\frac{\sqrt{\cosh\xi-\cos\eta}}{\sqrt{\cosh(\xi-2t)-\cos\eta}}.
\nonumber
\end{align}
By letting $a=e^{-|\xi-2t|}$ and $x=\cos\eta$ in \cref{gen_P0}, it is easy to check that
\begin{equation*}
\ds \frac{1}{\sqrt{\cosh (\xi-2t)-\cos\eta}}
=\sqrt{2}\sum_{n=0}^\infty e^{-\left(n+\frac{1}{2}\right)|\xi-2t|}P_n(\cos\eta).
\end{equation*}
Then, for $\xi < 2t$, we have
\beq\label{inv_distance_TObasis}
\ds\frac{\alpha}{\sinh |t|}\frac{1}{|\Br-\mathbf{z}(t)|}=
\sqrt{2}
\sqrt{\cosh\xi-\cos\eta} 
\sum_{n=0}^\infty e^{-\left(2n+1\right)t} e^{(n+\frac{1}{2})\xi}P_n(\cos\eta).
\eeq
Using the fact that $z_0=(0,0,\alpha \coth (s+t_0))\in B_+$, it can be shown that $t_0> -s/2$. It implies that $\xi < 2(ks+s+t_0)$ for $|\xi|\leq s$. Recall that $|\xi|\leq s$ for $\Br\in \mathbb{R}^3\setminus(B_+\cup B_-)$. Hence \cref{inv_distance_TObasis} holds for $t=ks+s+t_0$ and $\Br\in \mathbb{R}^3\setminus(B_+\cup B_-)$. Note that $\mathbf{z}_k=\mathbf{z}(ks+s+t_0)$. Therefore, using \eqref{TO_basis_bispherical}  and the definitions of $u_{k}$ and $G$, the conclusion follows for $u_k G(\Br-\mathbf{z}_k)$. The other case for $u_k G(\Br+\mathbf{z}_k)$ can be considered in the same way.

\section{Proof of Theorem \ref{main_thm1}} \label{app_pf_main_thm1}
We shall prove the result by applying our connection formula to Poladian's image series solution. From the image solution \cref{Poladian_sol_V}, Theorem \ref{thm_image_to_TO_U} and the identity $\partial_{z_0}\big|_{z_0=d} = -({\sinh^2 s}/{\alpha}) \partial_{t_0}\big|_{t_0=0}$, we get
\begin{align}
\ds V(\Br) + E_0 z &\approx V_1(\Br) + V_2(\Br)=4\pi p_0 {\partial_{z_0}}\big|_{z_0=d} U(\Br) + Q U(\Br)|_{z_0=d}
\nonumber
\\[0.5em]
\ds
&= E_0\frac{2\tau \alpha}{3-\tau} \sum_{n=0}^\infty \frac{2n+1 - \coth s}{e^{(2n+1)s}-\tau} 
(\mathcal{M}_{n,+}^0(\Br)-\mathcal{M}_{n,-}^0(\Br))
\label{V12_TO}
\\[0.5em]
\ds
& \quad+ Q\sum_{n=0}^\infty \frac{\mathcal{M}_{n,+}^0(\Br)-\mathcal{M}_{n,-}^0(\Br)}{e^{(2n+1)s}-\tau}.
\nonumber
\end{align}

Now let us consider the constant $Q$. Recall the following condition:
\begin{equation*}
\int_{\partial B_+} \frac{\partial}{\partial\mathbf{n}}  (V_1+V_2) \,dS = 0.
\end{equation*}
Then, by using Theorem \ref{thm_flux}, we obtain
\begin{equation*}
\ds E_0 \frac{2\tau \alpha}{3-\tau}\sum_{n=0}^\infty \frac{2n+1-\coth s }{e^{(2n+1)s}-\tau}+Q \sum_{n=0}^\infty \frac{1}{e^{(2n+1)s}-\tau}=0.
\end{equation*}
Hence, we see that $Q=-K_0+ E_0 \frac{2\tau \alpha}{3-\tau} \coth s$. Therefore, from \cref{V12_TO}, the conclusion follows.

\section{Proof of Theorem \ref{thm_multipole_TO}} \label{app_pf_thm_multipole_TO}
We first consider the case of a sectoral multipole $\mathcal{Y}_{|m|,m}$. We can represent the image potential $u_{m}^{(k)}\mathcal{Y}_{|m|,m}(\Br\mp \mathbf{z}_k)$ using TO basis as follows.
\begin{lemma}(Converting image sectoral multipole to TO)\label{lem_sectoral_to_TO}
For $\Br\in \mathbb{R}^3\setminus(B_+\cup B_-)$, we have
\begin{equation*}
\ds u_{m}^{(k)}\mathcal{Y}_{|m|,m}(\Br\mp \mathbf{z}_k) = \sum_{n=|m|}^\infty  g_{n}^m  \lambda_{n}^m \big[\tau e^{-(2n+1)s}\big]^k 
e^{-(2n+1)s} \mathcal{M}_{n,\pm}^m(\Br),
\end{equation*}
where $\lambda_n^m$ and $g_n^m$ are given in \cref{def_gnm_lambdanm_Q}.
\end{lemma}

\begin{proof}
For simplicity, we consider only $u_{m}^{(k)}\mathcal{Y}_{|m|,m}(\Br-\mathbf{z}_k)$. From  \cref{Pnn} and the fact that $\rho = |\mathbf{r}-\mathbf{z}_k|\sin \theta_k$, we have 
\begin{align}
\ds\mathcal{Y}_{|m|,m}(\Br- \mathbf{z}_k) &= \frac{1}{\sqrt{(2|m|)!}}\frac{P_{|m|}^{|m|}(\cos \theta_k)e^{im\phi_k}}{|\Br- \mathbf{z}_k|^{|m|+1}}
\label{CYmm_1}
\\
\ds
&= \omega_m\frac{ [\sin\theta_k]^{|m|} e^{im\phi_k} }{{|\Br- \mathbf{z}_k|^{|m|+1}}}
=\omega_m\frac{ \rho^{|m|} e^{im\phi_k} }{{|\Br- \mathbf{z}_k|^{2|m|+1}}},
\nonumber
\end{align}
where  $\omega_m={(-1)^{|m|}(2|m|-1)!!}/{\sqrt{(2|m|)!}}$ and $(r_k,\theta_k,\phi_k)$ is the spherical coordinates system for $\Br-\mathbf{z}_k$. Note that $\phi_k=\varphi$ for all $k\geq 0$.

From  \cref{inv_distance_bisph} and the fact that $\mathbf{z}_k = \mathbf{z}(ks+s+t_0)$, we see that
\begin{equation*}
\frac{1}{|\Br-\mathbf{z}_k|} =  \frac{{\sin(ks+s+t_0)}\sqrt{\cosh\xi-\cos\eta}}{\alpha\sqrt{\cosh(\xi- 2(ks+s+t_0))-\cos\eta}}.
\end{equation*}
We also have from  \cref{bisph_cartes_app} that
$
\rho= {\alpha \sin\eta}/{(\cosh\xi-\cos\eta)}.
$
By substituting these expressions for $1/|\Br-\mathbf{z}_k|$ and $\rho$ into \cref{CYmm_1}, we get
\begin{align}
\ds u_m^{(k)}\mathcal{Y}_{|m|,m}(\Br-\mathbf{z}_k)
&=\tau^k \frac{\sinh^{2|m|+1}(s+t_0)}{\sqrt{(2|m|)!}\alpha^{|m|+1}}\sqrt{\cosh\xi-\cos\eta}
 \label{umk_CYmm_1}
\\
\ds & \quad \times 
 \frac{ 2^{|m|+1/2}(-1)^{|m|}(2|m|-1)!! [\sin\eta]^{|m|} e^{im\phi_k}}{[2(\cosh(\xi-2(ks+s+t_0))-\cos\eta)]^{|m|+1/2}}.
 \nonumber
\end{align}
By letting $a=e^{-|\xi-2t|}$ and $x=\cos\eta$ in \cref{gen_Pnm}, it is easy to check that
\begin{equation*}
\ds\frac{(-1)^m(2m-1)!![\sin\eta]^m }{[2(\cosh(\xi-2t)-\cos\eta)]^{m+1/2}}
= \sum_{n=m}^\infty e^{-(n+\frac{1}{2})|\xi-2t|} P_n^m(\cos\eta).
\end{equation*}
By applying this identity to \cref{umk_CYmm_1} with $t=ks+s+t_0$, we obtain that
\begin{align*}
\ds u_m^{(k)}\mathcal{Y}_{|m|,m}(\Br-\mathbf{z}_k)
&=\tau^k 2^{|m|} \frac{\sinh^{2|m|+1}(s+t_0)}{\sqrt{(2|m|)!}\alpha^{|m|+1}}
\sqrt{2}\sqrt{\cosh\xi-\cos\eta}
\\
\ds & \quad \times 
\sum_{n=|m|}^\infty e^{-(2n+1)(ks+s+t_0)}e^{(n+\frac{1}{2})\xi} P_n^{|m|}(\cos\eta)e^{im\varphi},
\end{align*}
for $|\xi|\leq s$. Then, from \cref{TO_basis_bispherical}, the conclusion follows.
\end{proof}

\smallskip

Now we are ready to prove Theorem \ref{thm_multipole_TO}. By applying Lemma \ref{lem_sectoral_to_TO} to  \cref{Umm_image_sol} and then using the following identity:
\beq\label{tau_sum_geom2}
\sum_{k=0}^\infty \big[\tau e^{-(2n+1)s}\big]^{2k}  = \frac{e^{2(2n+1)s}}{e^{2(2n+1)s}-\tau^2},
\eeq
we obtain
\begin{equation*}
\ds U^\pm_{m}(\Br) =  \sum_{n=|m|}^\infty g_n^m \lambda_{n}^m  \frac{ e^{(2n+1)s}\mathcal{M}_{n,\pm}^m(\Br)-\tau\mathcal{M}_{n,\mp}^m(\Br) }{ e^{2(2n+1)s}-\tau^2 }.
\end{equation*}
Then, by using  \cref{ULM_image_sol_app}, we get
\begin{align}
\ds U^\pm_{lm}(\Br) &=  \sum_{n=|m|}^\infty   \frac{ g_n^m \mathcal{D}_{lm}^\pm[\lambda_{n}^m] }{ e^{2(2n+1)s}-\tau^2 }
(e^{(2n+1)s}\mathcal{M}_{n,\pm}^m(\Br)-\tau\mathcal{M}_{n,\mp}^m(\Br))\nonumber
\\
\ds &\qquad - \delta_{0m} Q_{l,1}^\pm \frac{\sinh s}{\alpha}\sum_{n=0}^\infty  \frac{  e^{(2n+1)s}\mathcal{M}_{n,+}^0(\Br)-\tau\mathcal{M}_{n,-}^0(\Br)}{ e^{2(2n+1)s}-\tau^2 }
\label{eq_Ulm_TO_1}
\\ 
\ds &\qquad - \delta_{0m} Q_{l,2}^\pm \frac{\sinh s}{\alpha}\sum_{n=0}^\infty  \frac{  (-\tau)\mathcal{M}_{n,+}^0(\Br)+e^{(2n+1)s}\mathcal{M}_{n,-}^0(\Br)}{ e^{2(2n+1)s}-\tau^2 }.
\nonumber
\end{align}

Now we consider the following flux conditions:
\begin{equation*}
\int_{\partial B_+} \frac{\partial (U_{l,m}^\pm)}{\partial\mathbf{n}} dS=0, \quad \int_{\partial B_-} \frac{\partial (U_{l,m}^\pm)}{\partial\mathbf{n}} dS=0.
\end{equation*}
Then, by applying Theorem \ref{thm_flux} to the above conditions with \cref{eq_Ulm_TO_1}, we obtain 
\begin{equation*}
 Q_{l,1}^\pm \frac{\sinh s}{\alpha}= \frac{\widetilde Q_{l,1}^\pm \pm \widetilde Q_{l,2}^\pm}{2}, \quad  Q_{l,2}^\pm \frac{\sinh s}{\alpha}= (-1)^l\frac{\widetilde Q_{l,1}^\pm \mp \widetilde Q_{l,2}^\pm}{2}.
\end{equation*}
By rearranging the terms, the proof is completed.

\section{Field at the gap center and absorption cross section}\label{app_field_gap_absorb}
From \cref{grad_origin_bisph}, we can see that the magnitude of the electric field at the gap is given by
\begin{equation*}
E=-(\nabla V \cdot {\mathbf{e}}_z) (0,0,0) =E_0 - \frac{2^{3/2}}{\alpha}\sum_{n=0}^\infty A_n (2n+1)(-1)^{n}.
\end{equation*}
As mentioned in the main text, the absorption cross section $\sigma_a$ is given by $\sigma_{a} = \omega \mbox{Im}\{ {p}\}$ where $p$ is the polarizability. It was shown in \cite{GN} that the polarizability $p$ is given by $p= \sqrt{2}\alpha^2 \sum_{n=0}^\infty (2n+1)A_n$. Therefore, by replacing $A_n$ by $\widetilde{A}_n$, we can derive  approximate analytical expressions for $E$ and $\sigma_a$.

\section{Multipole expansion method}\label{app_multipole}
The classical way to solve the many-spheres problem is Rayleigh's multipole expansion method. Here, we briefly review this method. Recall that the solid harmonics $\mathcal{Y}_{lm}$ and $\mathcal{Z}_{lm}$ are defined by
\begin{equation*}
\mathcal{Y}_{lm}(\Br)=\frac{Y_{l}^m(\theta,\phi)}{r^{l+1}}, \quad \mathcal{Z}_{lm}(\Br)=r^l Y_{l}^m(\theta,\phi).
\end{equation*}
Any solution to Laplace's equation can be represented as a sum of $\mathcal{Y}_{lm}$ and $\mathcal{Z}_{lm}$.
The solution $V(\Br)$ to the problem \cref{eqn:potentialV_trans} can be represented as the following multipole expansion: for $\Br$ belonging to the region outside the spheres, we have
\beq\label{multipole_exp}
V(\Br)=-E_0 z + \sum_{j=1}^J
\sum_{l=1}^\infty \sum_{m=-l}^l C_{j,lm}\mathcal{Y}_{lm}(\Br-\Bc_j),
\eeq
where the coefficients ${C}_{j,lm}$ are unknown constants and $\Bc_j$ is the center of the sphere $B_j$. For the inner region of $B_j$, we can easily extend the above representation by imposing the continuity of the potential on the surface $\partial B_j$. For $\Br \in B_j$, we have
\begin{equation*}
V(\Br)=\sum_{l=0}^\infty \sum_{m=-l}^l C_{j,lm}\frac{\mathcal{Z}_{lm}(\Br-\Bc_j)}{R^{2l+1}}.
\end{equation*}
Then, by using the addition formula \cref{translation_multi} for $\mathcal{Y}_{lm}$ and the transmission condition, $\nabla V\cdot\mathbf{n}|_{+} = \ep_B\, \nabla V\cdot\mathbf{n}|_{-}$ on the surface $\partial B_j$, the infinite dimensional linear system for unknowns $C_{j,lm}$ can be derived. If all the spheres are well-separated, good accuracy can be achieved by truncating the linear system by a small order.  But, if some of the spheres are close to touching, the charge densities on their surfaces become nearly singular. So more harmonics are required to describe them accurately.

\section{Hybrid numerical scheme for many plasmonic spheres}\label{sec:hybrid}
Here we present our hybrid numerical scheme for computing the field generated by plasmonic spheres clusters.
We modify Cheng and Greengard's hybrid method \cite{CG,C,Gan} by using our connection formulas between TO and the image method. So we first explain their idea in detail and then explain how we modify it for the plasmonic spheres system.

\subsection{Cheng and Greengard's hybrid method}
To illustrate the idea of the hybrid numerical schemes in \cite{CG,C,Gan}, let us consider an example of three  spheres (that is, $N=3$). Suppose that the spheres $B_1$ and $B_2$ are closely located but well-separated from $B_3$. Then the charge density on $\partial B_3$ can be well represented by a low-order spherical harmonics expansion. But the charge densities both on $\partial B_1$ and $\partial B_2$ may be singular, so it is better to use the image method  to describe their associated potentials. In view of this observation, they modified the multipole expansion as follows: for $\Br$ belonging to the region outside the spheres,
\begin{equation*}
\ds V(\Br)=-E_0 z +  \sum_{j=1}^2 \sum_{l=1}^\infty\sum_{m=-l}^l C_{12,lm} U_{12,lm}(\Br)
+ \sum_{l=1}^\infty\sum_{m=-l}^l C_{3,lm} \mathcal{Y}_{lm}(\Br-\Bc_3),
\end{equation*}
where $U_{12,lm}$ is the image series solution which includes all the image sources induced from the multipoles $C_{j,lm}\mathcal{Y}_{lm}(\Br-\Bc_j),j=1,2,$ by the interaction between two spheres $B_1$ and $B_2$. This representation for $V(\Br)$ can be directly generalized to a system of an arbitrary number of spheres. The resulting scheme is extremely efficient and accurate even if the spheres are nearly touching. This is because the close-to-touching interaction is already captured in the image multipole series.

\subsection{Outline of the modified algorithm for plasmonic spheres} 
As already mentioned, the image-series-based hybrid method cannot be applied for plasmonic spheres due to the non-convergence of the image series. Our strategy for extending the hybrid method to the case of plasmonic spheres is to convert the image multipole series  to their TO-type versions using the connection formula for general multipoles (Theorem \ref{thm_multipole_TO}). As a result, we obtain the modified hybrid numerical scheme valid for plasmonic spheres clusters. Here, we present the outline of the algorithm of our proposed scheme. 
\smallskip
\begin{itemize}
\item[1.] Write down the potential $V(\Br)$ in the multipole expansion form as in \cref{multipole_exp}. 
\item[2.] If a pair of spheres, say $B_j$ and $B_k$, are closely located (if the separation distance is smaller than a given number, for example, the radius $R$), then we rotate the $xyz$-axis for both $\Br-\Bc_j$ and $\Br-\Bc_k$ so that the $+z$-axis is in the direction of the axis of the pair of spheres, that is, $\Bc_j-\Bc_k$. 
\item[3.] We also transform the multipole expansion for $B_j$ into the rotated frame using  \cref{rotation_multi}. Let us denote the coefficients in the rotated frame by $C'_{j,lm}$.
\item[4.] By using the connection formula for general multipoles (Theorem \ref{thm_multipole_TO}), we modify the multipole expansion in the rotated frame by replacing $C'_{j,lm}\mathcal{Y}_{lm}(\Br)$ with the hybrid TO multipole $C'_{j,lm}U_{lm}^+(\Br)$. 
\item[5.] Do the same as in step 4 for $B_k$ with $U^-_{lm}(\Br)$ instead of $U^+_{lm}(\Br)$.
\item[6.] We convert the TO-type expansion for $B_j$ and $B_k$ into the form of multipole expansion using Theorem \ref{thm_conv_TO_to_Multi}. 
\item[7.] Rotate the axis of the coordinate system and transform the multipole expansions into the original frame. 
\item[8.] Perform steps 2-7 for all the pairs of closely spaced spheres. 
\item[9.] We extend the resulting multipole expansion to the inner regions of $B_j$ for $j=1,2,...,N$ using  Theorem \ref{inside_multi}.
\item[10.] By applying the addition formula \cref{translation_multi} for $\mathcal{Y}_{lm}$ with the transmission conditions on $\partial B_j$, we construct the infinite dimensional linear system for unknowns $C_{j,lm}$. 
\item[11.] We solve the linear system after a truncation.
\end{itemize}

\section{Useful formulas}\label{sec:useful}
Here we collect many useful formulas.

\subsection{Potential inside two spheres}

The following theorems are useful for finding a potential inside the two spheres when we have an explicit representation of the potential in the outside region. 

\begin{theorem}\label{inside_TO} 
Suppose that $V$ satisfies the Laplace equation inside and outside the two spheres $B_+$ and $B_-$. We also assume that the potential $V$ is continuous on each surface $\partial B_\pm$. 
We also assume that, outside the spheres, the potential $V$ is given by
\begin{equation*}
V(\Br)=\sum_{n=0}^\infty\sum_{m=-n}^n a_{n,+}^m \mathcal{M}_{n,+}^m(\Br) + a_{n,-}^m \mathcal{M}_{n,-}^m(\Br),
\end{equation*}
for $ \Br \in \mathbb{R}^3\setminus(B_+\cup B_-)$ and some coefficients $a_{n,\pm}^m$. Then,  inside the spheres, the potential $V(\Br)$ for $\Br\in B_\pm$ is given by
\begin{equation*}
V(\Br)=\ds\sum_{n=0}^\infty\sum_{m=-n}^n ({a}_{n,\pm}^m e^{(2n+1)s} +  a_{n,\mp}^m   ) \mathcal{M}_{n,\mp}^m(\Br),
\end{equation*}
\end{theorem}
\begin{proof}
It is obvious that the series on the right-hand side satisfies the Laplace equation. Since $\partial B_\pm = \{\xi=\pm s\}$, we have the following identity:
\begin{align*}
\ds \mathcal{M}_{n,+}^m(\Br) |_{\partial B_{\pm}} &= \sqrt{2}\sqrt{\cosh\xi-\cos\eta}e^{\pm(n+1/2)s}Y_n^{m}(\eta,\varphi)
\\
\ds & = e^{\pm(2n+1)s}\mathcal{M}_{n,-}^m(\Br)|_{\partial B_\pm}.
\end{align*}
Then one can easily check that the potential $V$ is continuous on each surface $\partial B_\pm = \{\xi=\pm s\}$. The proof is completed. 
\end{proof}

\begin{theorem} \label{inside_multi}
Suppose that $V$ satisfies the Laplace equation inside and outside the two spheres $B_+$ and $B_-$. We also assume that the potential $V$ is continuous on each surface $\partial B_\pm$. Furthermore, we assume that, outside the spheres, the potential $V$ is given by
\begin{equation*}
V(\Br)=\sum_{l=0}^\infty\sum_{m=-l}^l f_{lm}^+ \mathcal{Y}_{lm}(\Br-\Br_0) + f_{lm}^- \mathcal{Y}_{lm}(\Br+\mathbf{r}_0),
\end{equation*}
for $ \Br \in \mathbb{R}^3\setminus(B_+\cup B_-)$ and some coefficients $f_{l,m}^\pm$. Then, inside the spheres, the potential $V(\Br)$ for $\Br\in B_\pm$ is given by
\begin{equation*}
V(\Br)=\sum_{l=0}^\infty\sum_{m=-l}^l \frac{f_{lm}^+}{R^{2l+1}} \mathcal{Z}_{lm}(\Br-\Br_0) + \frac{f_{lm}^-}{R^{2l+1}} \mathcal{Z}_{lm}(\Br+\mathbf{r}_0).
\end{equation*}

\end{theorem}
\begin{proof}
The conclusion immediately follows from the definition of the solid harmonics $\mathcal{Y}_{lm}$ and $\mathcal{Z}_{lm}$.
\end{proof}

\subsection{From TO to multipole expansion}

When we apply our hybrid numerical scheme for plasmonic spheres, we need to convert a TO-type solution into a multipole expansion. Let us consider the following  general potential $W_\pm$ in the form of TO solution:
\beq\label{general_W_TO}
W_\pm (\Br)=
\ds\sum_{n=0}^\infty\sum_{m=-n}^n a_{n,\pm}^m \mathcal{M}_{n,\pm}^m(\Br),
\eeq
for some coefficients $a_{n,\pm}^m$. We want to convert the potential $W_\pm$ into a multipole expansion form as follows:
\beq\label{general_W_multi}
W_\pm(\Br) =
\begin{cases}
\ds\sum_{l=0}^\infty \sum_{m=-n}^n c_{lm}^\pm \mathcal{Y}_{lm}(\Br \mp \Br_0),
&\quad \Br \in \mathbb{R}^3\setminus B_\pm,
\\
\ds\sum_{l=0}^\infty \sum_{m=-n}^n d_{lm}^\pm \mathcal{Z}_{lm}(\Br\mp\Br_0), 
&\quad \Br \in  B_\pm,
\end{cases}
\eeq
where the coefficients $c^\pm_{lm}$ and $d^\pm_{lm}$ are to be determined. We derive explicit formulas for $c^\pm_{lm}$ and $d_{lm}^\pm$ in the following theorem. Its proof is given in Appendix \ref{app_pf_TO_Multiplole}.

\begin{theorem}\label{thm_conv_TO_to_Multi} (Conversion of TO solution into multipole expansion)
The multipole coefficients $c_{lm}^\pm$ and $d_{lm}^\pm$ are represented in terms of TO coefficients $a_{n,\pm}^m$ as follows:
\begin{equation*}
\begin{cases}
\ds c_{lm}^\pm = 2\alpha R^{2l+1}\sum_{n=|m|}^\infty a_{n,\pm}^m  g_{n}^m   \mathcal{D}_{lm}^\pm [\lambda_{n}^m],
\\
\ds d_{lm}^\pm = 2\alpha \sum_{n=|m|}^\infty a_{n,\mp}^m e^{-(2n+1)s} g_{n}^m   \mathcal{D}_{lm}^\pm [\lambda_{n}^m].
\end{cases}
\end{equation*}
\end{theorem}

In view of \cref{general_W_multi}, the total flux on the surface $\partial B_\pm$ is given as
\begin{equation*}
\int_{\partial B_\pm } \frac{\partial W_\pm}{\partial \mathbf{n}} dS = 4\pi c_{00}^\pm, \quad \int_{\partial B_\pm } \frac{\partial W_\mp}{\partial \mathbf{n}} dS = 0.
\end{equation*}
So, we have the following flux formula from the above theorem.

\begin{theorem}\label{thm_flux}
(Total flux formula) Let $W_\pm$ be the potential given as \cref{general_W_TO}. Then the total flux on the surface $\partial B_\pm$ is 
\begin{equation*}
\ds\int_{\partial B_\pm } \frac{\partial W_\pm}{\partial \mathbf{n}} dS = 8\pi \alpha \sum_{n=0}^\infty a_{n,\pm}^0 
,\quad
\int_{\partial B_\pm } \frac{\partial W_\mp}{\partial \mathbf{n}} dS = 0.
\end{equation*}

\end{theorem}

\subsection{Translation and rotation of multipole expansions}

To apply the multipole expansion method,  we need  to represent a multipole source in a translated or rotated coordinate. It was shown that the following identities hold \cite{Pol_thesis}.
\smallskip

\noindent\textbf{Translation}:
A translated multipole source $\mathcal{Y}_{lm}(\Br-\Br')$ can be expanded as
\beq\label{translation_multi}
\ds \mathcal{Y}_{lm}(\Br-\Br') = \sum_{a=0}^\infty\sum_{b=-a}^a  w_m w_b w_{m-b} 
 N_{lmab}(-1)^{l+a} \mathcal{Z}_{ab}(\Br_{<})\mathcal{Y}_{l+a,m-b}(\Br_>),
\eeq
where $\Br_{<}$ is the smaller (in magnitude) of $\Br$ and $\Br'$ and $\Br_{>}$ is the larger. Here $w_{m}$ and $N_{lmab}$ are defined as \cref{def_wm_app} and \cref{def_Nlmab_app}, respectively.

\smallskip
\noindent\textbf{Rotation}:
Suppose that the coordinate axes are rotated through Euler angle $\alpha,\beta,\gamma$. The point $(\theta,\phi)$ becomes $(\widetilde{\theta},\widetilde{\phi})$. The following result holds:
\beq\label{rotation_multi}
Y_{lm}(\theta,\phi) = \sum_{M=-l}^l w_m w_M D^{(l)}_{mM}(\alpha,\beta,\gamma) Y_{l}^M(\widetilde{\theta},\widetilde{\phi}),
\eeq
where $D^{(l)}_{mM}(\alpha,\beta,\gamma)=e^{-i\alpha +M\gamma} d_{mM}^l(\beta)$ and
\begin{align*}
\ds d_{mM}^l(\beta) &= \cos(\beta/2)^{2l+m-M}\sin(\beta/2)^{M-m}
\\
\ds\quad \qquad\quad &\times \sum_{t}
\sqrt{     
\begin{pmatrix}
l+m \\ t
\end{pmatrix}
\begin{pmatrix}
l-M \\ t
\end{pmatrix}
\begin{pmatrix}
l+M \\ l+m-t
\end{pmatrix}
\begin{pmatrix}
l-m \\ l-M-t
\end{pmatrix}
}
\\
\ds \quad\qquad\quad &\times 
(-1)^t \tan(\beta/2)^{2t}.
\end{align*}
The summation in $t$ is carried over $\max (0,m-M) \leq t \leq \min (l+m,l-M)$.

\section{Proof of Theorem \ref{thm_conv_TO_to_Multi}} \label{app_pf_TO_Multiplole}
Let $\sigma_\pm $ be the charge density on the surface $\partial B_\pm$, respectively. Now let us decompose $\sigma_\pm$ using the spherical harmonics $Y_l^m(\theta_\pm,\phi_\pm)$, where $(r_\pm,\theta_\pm,\phi_\pm)$ are the spherical coordinates for $\Br\mp \Br_0$, respectively. Let us write  $\sigma_\pm$ as
\begin{equation*}
\sigma_\pm  = \sum_{l=0}^\infty \sum_{m=-l}^{l} \sigma_{lm}^\pm Y_{l}^m(\theta_\pm,\phi_\pm).
\end{equation*}
Here, $\sigma_{lm}^\pm$ can be determined using the orthogonality of the spherical harmonics as follows:
\beq\label{sigma_pm_lm_ortho}
\ds {\sigma}_{lm}^\pm 
= \frac{2l+1}{4\pi}\frac{1}{R^2}\int_{\partial B_\pm} \sigma_\pm \overline{Y_{l}^m}(\theta_\pm,\phi_\pm) \,dS. 
\eeq
To calculate the right-hand side of  (\ref{sigma_pm_lm_ortho}), we need to express $\sigma_\pm$ and $Y_{l}^m(\theta_\pm,\phi_\pm)$ in terms of TO harmonics $Y_{n}^m(\eta,\varphi)$.

First, let us consider $\sigma_\pm$. Let 'ext'(or 'int') denote the limit from the outside (or inside) the sphere, respectively. It is well-known that the electric field $\mathbf{E}=-\nabla W$ satisfies the following boundary condition on $\partial B_\pm$:
\begin{equation*}
\mathbf{E}\cdot \mathbf{n}|_{ext} - \mathbf{E}\cdot\mathbf{n}|_{int} = \sigma_\pm, \quad \mbox{on }\partial B_\pm.
\end{equation*}
To use the above condition, we need an explicit expression for $W_\pm$ in the region inside the spheres $B_\pm$, respectively. From Theorem \ref{inside_TO}, we have, for $\Br \in B_\pm$,  
\beq
W_\pm(\Br)=
\ds\sum_{n=0}^\infty\sum_{m=-n}^n {a}_{n,\pm}^m e^{(2n+1)s}    \mathcal{M}_{n,\mp}^m(\Br),
\eeq
respectively.
So, by using  \cref{normal_deri}, we obtain
\begin{align}
\ds \sigma_{\pm} &= -\frac{\partial W}{\partial\mathbf{n}}\Big|^{ext}_{\partial B_+}+\frac{\partial W}{\partial\mathbf{n}}\Big|^{int}_{\partial B_+}\label{sigmapm_TO}
\\
\ds &= (2\alpha)^{1/2} [J(\eta)]^{-3/2}
\sum_{n,m} a_{n,\pm}^m (2n+1)  e^{(n+\frac{1}{2})s}Y_n^m(\eta,\varphi),
\nonumber
\end{align}
where $J(\eta)$ is defined by $J(\eta) = {\alpha}/{(\cosh s-\cos\eta)}$.

Next, let us consider $Y_{l}^m(\theta_\pm,\phi_\pm)$. From  \cref{Ylm_and_Ymm} and Lemma \ref{lem_sectoral_to_TO}, we have for $\Br\in \partial B_+$,
\begin{align}
\ds Y_{l}^m(\theta_\pm,\varphi_\pm) &= R^{l+1}\mathcal{Y}_{lm}(\Br\mp\Br_0)
=R^{l+1}\mathcal{D}_{lm}^\pm[\mathcal{Y}_{|m|,m}(\Br\mp \mathbf{z}_0)]
\label{Ylm_th_ph_TO}
\\
\ds &=R^{l+1}
(2\alpha)^{1/2}[J(\eta)]^{-1/2}
\sum_{n=|m|}^\infty g_n^m \mathcal{D}_{lm}^\pm[\lambda_{n}^m] e^{-(n+1/2)s}Y_{n}^m(\eta,\varphi).
\nonumber
\end{align}

We are ready to compute $\sigma_{lm}^+$. By substituting  \cref{sigmapm_TO}  and \cref{Ylm_th_ph_TO}  into \cref{sigma_pm_lm_ortho}, we obtain
\begin{align}
\ds{\sigma}_{lm}^\pm 
&= \frac{2l+1}{4\pi }\frac{1}{R^2}\int_{0}^{2\pi}\int_{0}^\pi \sigma_\pm \overline{Y_{lm}} [J(\eta)]^2 \sin\eta \, d \eta \,d\varphi
\label{sigma_lm_pm}
\\
&={(2l+1)} {2\alpha R^{l-1}} \sum_{n=|m|}^\infty a_{n,\pm}^m  g_n^m   \mathcal{D}_{lm}^\pm [\lambda_{n}^m].
\nonumber
\end{align}

It is easy to check that the potential generated by the charge densities $\sigma_\pm=\sum \sigma_{lm}^\pm Y_{l}^m$ is given as follows: for $\Br\in \mathbb{R}^3\setminus(B_+\cup B_-)$,
\begin{equation*}
W_\pm (\Br)= 
\ds\sum_{l,m} \sigma_{lm}^\pm \frac{R^{l+2}}{2l+1} \mathcal{Y}_{lm} (\Br\mp \Br_0). 
\end{equation*}
By comparing the above expression and  \cref{general_W_multi}, we immediately arrive at
\begin{equation*}
\ds c_{lm}^\pm = \sigma_{lm}^\pm\frac{R^{l+2}}{2l+1}.
\end{equation*}
Then, the formula for $c_{lm}^\pm$ follows from \cref{sigma_lm_pm}. For the case of $d_{lm}^\pm$, it can be proved in a similar way.

\section*{Acknowledgments}
The authors would like to thank Ross C. McPhedran and Graeme W. Milton for pointing out the existence of Poladian's thesis \cite{Pol_thesis}.

\bibliographystyle{siamplain}

\end{document}